\documentclass[12pt,a4paper]{article}

\usepackage[a4paper,margin=1in]{geometry}
\usepackage{amsmath,amssymb,bm}
\usepackage{amsthm}
\usepackage{booktabs}
\usepackage{graphicx}
\usepackage{float}
\usepackage{placeins}
\usepackage{enumitem}
\usepackage[numbers,square,sort&compress]{natbib}
\usepackage[nopatch=footnote]{microtype}
\usepackage{xcolor}
\usepackage{setspace}
\usepackage{hyperref}
\newcommand{\R}{\mathbb R}

\newcommand{\one}{\bm 1}

\newcommand{\argmin}{\operatorname*{arg\,min}}
\newtheorem{proposition}{Proposition}

\theoremstyle{definition}

\theoremstyle{remark}
\newtheorem{remark}{Remark}
\theoremstyle{plain}

\newcommand{\cD}{\mathcal D}
\newcommand{\cH}{\mathcal H}
\newcommand{\cF}{\mathcal F}
\newcommand{\ve}{\operatorname{vec}}
\hypersetup{
colorlinks=true,
linkcolor=blue!60!black,
urlcolor=blue!60!black,
citecolor=blue!60!black
}

\title{\Large\bfseries
Neighbourhood-Based Generalized Dynamic Principal Components for Spatial Functional Data}
\author{}
\date{}

\begin{document}

\maketitle

\vspace{-1.5em}

\begin{center}

\begin{minipage}{0.95\textwidth}
\centering

{\large
Dharini Pathmanathan\textsuperscript{a,b,*},
Teck Xiang Seow\textsuperscript{a},
Philipp Otto\textsuperscript{c},
and Sophie Dabo-Niang\textsuperscript{d}
}

\vspace{1.2em}

\small
\raggedright

\textsuperscript{a}Institute of Mathematical Sciences, Faculty of Science,
Universiti Malaya, 50603 Kuala Lumpur, Malaysia

\vspace{0.35em}

\textsuperscript{b}Centre of Research for Statistical Modelling and Methodology,
Faculty of Science, Universiti Malaya, 50603 Kuala Lumpur, Malaysia

\vspace{0.35em}

\textsuperscript{c}School of Mathematics and Statistics, University of Glasgow,
University Place, Glasgow, G12 8QQ, UK

\vspace{0.35em}

\textsuperscript{d}Univ. Lille, CNRS, UMR 8524 -- Laboratoire Paul Painlev\'e,
Inria-Datavers, F-59000 Lille, France

\vspace{0.6em}

\textsuperscript{*}\textit{Corresponding author:}
\href{mailto:dharini@um.edu.my}{dharini@um.edu.my}

\end{minipage}

\end{center}

\vspace{1.2em}

\singlespacing

\begin{abstract}
Regular-grid spatial functional datasets arise naturally in gridded environmental, oceanographic, climate, and remote-sensing applications, where each spatial location is associated with an entire curve. Existing spectral spatial functional principal component analysis methods provide an important frequency-domain approach for such data, but they do not directly target finite-neighbourhood least-squares reconstruction from an estimated latent spatial component field. To address this, we propose Spatial Functional Generalized Dynamic Principal Components (SFGDPC), a reconstruction-based dimension-reduction method for regular-grid spatial functional data. Each function is first represented by basis coefficients, and each coefficient vector is reconstructed from a scalar latent spatial field and its Chebyshev neighbourhood on the rectangular grid. The spatial neighbourhood radius is selected using a Bayesian information criterion (BIC)-type conditional reconstruction criterion. In the local-neighbourhood transfer simulation, SFGDPC reduced mean cumulative normalized mean squared error (NMSE) relative to spatial functional principal component analysis (SFPCA) by approximately \(38\)--\(53\%\) across the reported component counts and covariance conditions. In the Indian Ocean sea surface temperature (SST) application, one SFGDPC component produced lower whole-grid reconstruction error than both the boundary-safe and high-cap SFPCA benchmarks across all 33 annual fields. The results support SFGDPC as a local reconstruction-based complement to spectral SFPCA for regular-grid spatial functional data.
\end{abstract}


\section{Introduction}

Many modern environmental datasets consist not of a single measurement at each location, but of an entire temporal or spectral profile. Sea surface temperature (SST) over the course of a year \citep{reynolds_etal_2007,kuenzer_sfpca}, rainfall intensity across seasons \citep{delicado2010spatial,mateu2021geostatistical}, and remotely sensed vegetation trajectories \citep{liu2012remote} are typical examples. The resulting data combine two sources of structure: variation within each curve and dependence across spatial locations. Effective dimension reduction should therefore account for both. Such data are commonly described as a spatial functional process
\(\{X_s:s\in D\subset\mathbb{R}^r\}\), 
where $s$ denotes a spatial location, $r$ is the spatial dimension, and $X_s \in L^2(\mathcal T)$ is a square-integrable functional random variable, typically written as $X_s(t)$ for $t\in\mathcal{T}$ \citep{delicado2010spatial,mateu2021geostatistical,kuenzer_sfpca}. The present work focuses on the two-dimensional regular-grid case (\(r=2\)), with observations on a complete finite rectangular lattice, which we index as
\(D_N=\{1,\ldots,n_{1}\}\times\{1,\ldots,n_{2}\}, \; N=n_{1}\times n_{2}\), where \(n_{1}\) and \(n_{2}\) denote the numbers of grid rows and columns, respectively.

Regular-grid environmental and reanalysis products are a motivating data class for this work, as they often contain repeated functional profiles over rectangular spatial lattices. The SST application in Section~\ref{sec:sst} provides a representative empirical assessment of local reconstruction for this type of data. Several approaches have been developed for spatial functional data. Spectral/filter-based spatial functional principal component analysis (SFPCA) methods incorporate spatial dependence through frequency-domain structure under weak stationarity and have a well-developed asymptotic theory \citep{kuenzer_sfpca}. This spectral framework has also been extended to multivariate functional responses on regular grids \citep{siahmed2025smfpca}; the present paper concerns a univariate functional response per location, for which SFPCA \citep{kuenzer_sfpca} is the natural spectral benchmark.

A different line of work is model-based and geostatistical. Hidden dynamic geostatistical models (HDGMs) place a latent dynamic process on spatially indexed data and estimate the resulting model by likelihood- or expectation--maximization (EM)-type methods \citep{calculli_etal_2015}. Functional hidden dynamic geostatistical models (f-HDGMs) extend this idea to functional responses through basis-coefficient representations and are particularly suited to kriging, covariate adjustment, and irregularly located observations \citep{wang_finazzi_fasso_2021,maranzano_otto_fasso_2023}. Functional geostatistical prediction for Hilbert-space-valued data has also been studied in~\citep{MenafoglioSecchiDallaRosa2013}. These methods are model-based and inferential in nature, whereas the method developed here is descriptive and reconstruction-oriented.

For large regular-grid functional datasets, an important practical goal is to summarize the observed field using a small number of components while retaining sufficient information for reconstruction. The present work is motivated by a reconstruction question complementary to the spectral SFPCA formulation \citep{kuenzer_sfpca}: can the functional field be represented using a latent component field over a bounded grid neighbourhood, with the size of that neighbourhood chosen as part of the representation? In SFPCA, \(L\) is an integer-valued truncation parameter that bounds the spatial lags included in the estimated filters, whereas the neighbourhood radius considered here defines the spatial support of the reconstruction model. Generalized dynamic principal components (GDPC) \citep{pena_yohai_gdpc}, implemented in the \texttt{gdpc} R package \citep{pena_smucler_yohai_2020},
form a natural reconstruction-based starting point, but the method is
formulated for vector time series with a one-dimensional lag structure. These considerations motivate extending its reconstruction principle to functional observations on a two-dimensional regular grid.

The central contribution of this paper is Spatial Functional Generalized Dynamic Principal Components (SFGDPC). For each candidate radius, SFGDPC alternates between estimation of the latent spatial field and the functional loadings. The radius determines the maximum spatial offset included in the reconstruction, and its component-specific value is selected from \(k=0,\ldots,s_{\max}\) using a Bayesian information criterion (BIC)-type penalized reconstruction criterion. In this way, SFGDPC retains the least-squares reconstruction principle of GDPC while adapting its component and loading structure to a regular spatial lattice. The resulting framework links functional dimension reduction with finite-neighbourhood spatial reconstruction.

We evaluate the proposed method in three complementary settings: two simulation studies and one environmental application. The first is a stationary spatial functional autoregressive moving-average (SFARMA) benchmark favourable to spectral methods, used to examine whether SFGDPC remains competitive under a standard stationary spatial-functional process. The second is a local-neighbourhood latent-driver simulation, representing the regime targeted by the proposed method. The third is an application to Indian Ocean SST anomaly fields, used to assess empirical reconstruction performance and the spatial radii selected by SFGDPC.

The remainder of the paper is structured as follows. Section~\ref{sec:methodology} develops the model, estimation algorithm, and
BIC-type radius-selection criterion. Section~\ref{sec:simulation} reports the simulation studies, and Section~\ref{sec:sst} presents the SST application. Sections~\ref{sec:discussion} and~\ref{sec:conclusion} contain the discussion and conclusion, respectively.

\section{Methodology}
\label{sec:methodology}

Let $\{X_s:s\in\mathbb Z^2\}$ be a functional spatial process with values in
$L^2(\mathcal T)$. The observations are located on the complete rectangular
lattice $\mathcal D=\{1,\ldots,n_1\}\times\{1,\ldots,n_2\}$, with
$N=n_1n_2$. All observations are approximated in a common, fixed, $m$-dimensional basis
$\{\phi_1,\ldots,\phi_m\}$:
\[
  X_s^{(m)}(t)=\sum_{j=1}^m z_{sj}\phi_j(t),
  \qquad z_s=(z_{s1},\ldots,z_{sm})^\top\in\R^m.
\]
The basis may consist of functional principal components, splines, or Fourier functions. Let \(Z=(z_s^\top:s\in\cD)\in\R^{N\times m}\) denote the coefficient matrix after any analysis-specific centring. Its rows follow the ordering \(i(r,c)=c+(r-1)n_2\), for \(1\leq i(r,c)\leq N\).

\subsection{Spatial neighbourhood}

For an integer radius $k\geq0$, define the full Chebyshev offset set
$\cH_k=\{-k,\ldots,k\}^2$, with
$q_k=|\cH_k|=(2k+1)^2$ and $c_k=q_k-1$.
For each site $i$, the noncentral offsets in $\cH_k$ are scanned in row-major order, with the row offset varying slowest. Offsets whose translated sites lie inside $\cD$ are retained and compacted into the leading slots. If
fewer than $c_k$ valid neighbours exist, the remaining trailing slots are filled with the target index itself.
Writing $\nu_{k,0}(i)=i$ and letting
$\nu_{k,1}(i),\ldots,\nu_{k,c_k}(i)$ denote the resulting compacted and
self-padded indices, define the slot-selection operators
$\Pi_{k,r}\in\R^{N\times N}$ by
\begin{equation}
  (\Pi_{k,r}f)_i=f_{\nu_{k,r}(i)},
  \qquad r=0,\ldots,c_k.
  \label{eq:slot-selection}
\end{equation}
No index is mapped across a lattice row boundary, and the number of design columns is the same at every site.
The spatial design matrix is then
\begin{equation}
  F_k(f)
  =
  \bigl[
    \Pi_{k,0}f,\ldots,\Pi_{k,c_k}f,\one_N
  \bigr]
  \in\R^{N\times(q_k+1)} \, ,
  \label{eq:spatial-design}
\end{equation}
so that the intercept occupies the final column.

\subsection{Spatial functional generalized dynamic component}

For a fixed radius $k$, one component comprises a spatial score field $f=(f_i:i\in\cD)^\top\in\R^N$, an intercept $\alpha\in\R^m$, and loading
vectors $\beta_{k,r}\in\R^m$, $r=0,\ldots,c_k$. The coefficient vector at
site $i$ is reconstructed as
\begin{equation}
  \widehat z_i
  =
  \alpha+
  \sum_{r=0}^{c_k}
  \beta_{k,r}f_{\nu_{k,r}(i)}.
\label{ref3}
\end{equation}

Let
\[
  B_k=
  [\beta_{k,0},\ldots,\beta_{k,c_k}]
  \in\R^{m\times q_k},
  \qquad
  C_k=
  \begin{bmatrix}
    B_k^\top\\
    \alpha^\top
  \end{bmatrix}
  \in\R^{(q_k+1)\times m}.
\]
The fitted coefficient matrix is
\begin{equation}
  \widehat Z_k=F_k(f)C_k.
\label{ref4}
\end{equation}
Equation~(\ref{ref3}) is a reconstruction parameterization and need not be
assumed to be the true data-generating model.

At interior sites, slot $0$ represents the centre, while each slot
$r=1,\ldots,c_k$ corresponds to the $r$-th noncentral offset in the fixed
row-major ordering. Under this ordering, $\beta_{k,r}$ has its usual
lag-specific interpretation at interior sites.

At boundary sites, invalid offsets are removed and the surviving neighbours are compacted. The slot indices do not retain a universal fixed-offset interpretation. The trailing self-padded slots contribute an additional effective centre loading $\sum_{r\in\mathcal P_i}\beta_{k,r}$, where
$\mathcal P_i$ is the set of padded slots at site $i$. Lag-specific
interpretation of the fitted loadings is exact only at interior
sites. The same boundary convention is used in estimation, reconstruction, and model selection. The fitted object stores the intercept separately and the loading columns in reverse design order; the latter are reordered before reconstruction or interpretation.

For a nondegenerate score field, the implementation removes translation and
scale indeterminacy by imposing
\begin{equation}
  \one_N^\top f=0,
  \qquad
  \|f\|_2^2=N-1,
  \label{ref5}
\end{equation}
so that $f$ has zero mean and unit sample standard deviation. Its sign remains
arbitrary. The handling of a degenerate score update is described below.

For fixed $k$, SFGDPC seeks to minimize the reconstruction loss
\begin{equation}
  \mathcal L_k(f,C)
  =
  \frac{1}{Nm}
  \left\|Z-F_k(f)C\right\|_F^2
  \label{ref6}
\end{equation}
over score fields satisfying~(\ref{ref5}) and
$C\in\R^{(q_k+1)\times m}$.
We use \((\widehat f_k,\widehat C_k)\) to denote the output of the alternating least-squares algorithm described below; global minimization
of~(\ref{ref6}) is not guaranteed.

Only the fitted matrix and the attained loss are essential for reconstruction.
For a fixed score field, the least-squares coefficient matrix \(C\) is unique if and only if \(F_k(f)\) has full column rank. Here
$\|Z-F_k(f)C\|_F$ denotes the Frobenius norm.

The methodology then minimizes an ordinary Frobenius loss in coefficient space. This loss equals integrated squared functional error when the basis is orthonormal. For a basis with Gram matrix \(G_{ab}=\int_{\mathcal T}\phi_a(t)\phi_b(t)\,dt\), exact \(L^2\) fitting requires either transformed coefficients \(z_s^\star=G^{1/2}z_s\) or a Gram-weighted loss. In practice we use the unweighted coefficient-space loss.

\subsection{Finite-sample existence and the conditional loading update}

The finite-sample SFGDPC methodology does not require second-order
stationarity or isotropy. The adoption of a common loading vector for each neighbourhood slot across sites constitutes a structural constraint of the chosen reconstruction parameterization rather than an assumption of stationarity. Any stationarity assumptions used below pertain to specific data-generating mechanisms or competing methods and are not requirements of the SFGDPC estimator.

The following finite-sample conditions provide a sufficient theoretical setting for the existence of a minimizer of the restricted fixed-radius problem over $\mathcal F_k$ and for the uniqueness of the conditional loading update. Computational handling outside the full-rank setting is described subsequently.

\noindent\textbf{(A1) Common representation and finite data.}
The basis functions are linearly independent, $m<\infty$, and every entry of $Z$ is finite. The data contain no missing values. If normalized error is reported, $\|Z\|_F>0$.

\noindent\textbf{(A2) Valid lattice and candidate radii.}
The observed sites form the complete rectangle $\cD$ with known dimensions
$n_1$ and $n_2$, where $N=n_1n_2\geq2$. The vectorization is as specified
above, and $\mathcal K=\{0,\ldots,s_{\max}\}$ for a fixed finite integer
$s_{\max}\geq0$. For each $k\in\mathcal K$, $q_k+1\leq N$.

\noindent\textbf{(A3) Stable score designs.}
For each \(k\in\mathcal K\), let \(\cF_k\) be a nonempty compact set of score fields satisfying~(\ref{ref5}). For the finite-sample existence result,
assume that there exists \(\eta_k>0\) such that
\begin{equation}
\lambda_{\min}\{F_k(f)^\top F_k(f)\}\geq\eta_k
\qquad\text{for every }f\in\cF_k.
\label{ref7}
\end{equation}
This sufficient condition excludes rank-degenerate score configurations and
ensures that the least-squares coefficient matrix is uniquely defined and continuous in $f$.
  \begin{proposition}
  Under (A1)--(A3), for fixed $f\in\cF_k$, the
unique least-squares coefficient matrix is
\begin{equation}
  \widehat C_k(f)
  =\{F_k(f)^\top F_k(f)\}^{-1}F_k(f)^\top Z.
  \label{ref8}
\end{equation}

Moreover, the profiled loss $f\mapsto\mathcal L_k\{f,\widehat C_k(f)\}$ is continuous on $\cF_k$ and attains its minimum. Hence the restricted problem
\[
\min_{\substack{f\in\cF_k\\
C\in\R^{(q_k+1)\times m}}}
\mathcal L_k(f,C)
\]
has at least one solution.
  
  \end{proposition}
\textbf{Proof}:
\textit{Condition~(\ref{ref7}) makes $f\mapsto\widehat C_k(f)$ continuous. Consequently,
$f\mapsto\mathcal L_k\{f,\widehat C_k(f)\}$ is continuous on the compact set
$\cF_k$ and attains its minimum. Therefore, the restricted problem over $\cF_k$ has at least one solution. The sign of any minimizing score field may be fixed by any deterministic convention if a unique orientation is desired.}

\begin{remark}
Assumption~(A3) is an idealized sufficient condition for the existence and continuity result and is not verified by the implementation. For a fixed score, the implementation attempts a symmetric positive-definite solution of the normal equations and uses a Moore--Penrose inverse if that solve fails or returns nonfinite values. Outside the full-rank setting of the Proposition,
the implemented loading update is represented by $\widehat C_k(f)=F_k(f)^+Z$,
which coincides with~(\ref{ref8}) when $F_k(f)$ has full column rank.
This distinguishes full-rank uniqueness from computational well-definedness under the pseudoinverse fallback. The condition
$q_k+1\leq N$ is retained because it is necessary for possible full-column-rank
uniqueness, although it is not required merely to compute a pseudoinverse solution.
\end{remark}

\begin{remark}
If the centred raw score has nonfinite or near-zero sample standard deviation, the implementation returns the zero field rather than reinitializing it. The convergence flag indicates only that the algorithm terminated before the iteration limit and does not certify that the score is nondegenerate. The
selection criterion continues to use the nominal parameter count \(p_k\) for
such a candidate.
\end{remark}

\subsection{Alternating least-squares computation}

For a current score $f$, the intercept and loading matrix are updated using~(\ref{ref8}). For fixed $\alpha$ and $B_k$, define
\[
  A_k(B_k)
  =
  \sum_{r=0}^{c_k}
  \left(\beta_{k,r}\otimes\Pi_{k,r}\right)
  \in\R^{Nm\times N}.
\]
The unconstrained least-squares score update is characterized by
\begin{equation}
  f^{\mathrm{raw}}
  =
  A_k(B_k)^+
  \ve\!\left(Z-\one_N\alpha^\top\right),
  \label{ref9}
\end{equation}
where $A^+$ denotes the Moore--Penrose inverse. The implementation does not
form $A_k(B_k)^+$ explicitly. Instead, it accumulates the equivalent dense
$N\times N$ normal equations using the centre and neighbour slots
$\nu_{k,r}(i)$ and solves them directly. If the direct solve fails or returns
nonfinite values, the Moore--Penrose inverse of the normal-equation matrix is
used.

When the centred raw score has finite sample standard deviation exceeding \(10^{-12}\), it is centred and rescaled according to
\begin{equation}
  f\leftarrow \sqrt{N-1}\,
  \frac{M_Nf^{\mathrm{raw}}}
       {\|M_Nf^{\mathrm{raw}}\|_2},
  \qquad
  M_N=I_N-N^{-1}\one_N\one_N^\top.
  \label{ref10}
\end{equation}

If the sample standard deviation is nonfinite or no greater than \(10^{-12}\), the implementation instead returns the zero field.

Unless an initial score field is supplied, the current working coefficient matrix is centred across grid locations and its first right singular vector is computed. The corresponding score vector is then formed and standardized
according to~(\ref{ref5}). At component step $\ell$, the working matrix is $E^{(\ell-1)}$, with $E^{(0)}=Z$. This initialization is independent of $k$, so all candidate radii within the same component step use the same initial score field.

Since the intercept is unrestricted and the loading matrix is unconstrained in scale, this normalization does not change the minimum attainable reconstruction loss: its effect can be absorbed by adjusting the intercept
and rescaling the loadings. Since the loading and intercept parameters are refitted after the normalized score update, the loss evaluated after each complete pair of exact updates is non-increasing.

\begin{remark}
The updates stop when the relative change between successive mean squared error (MSE) values is no
larger than the prescribed numerical tolerance, or when the iteration limit is reached. This relative change is the absolute change in MSE divided by the absolute previous MSE plus machine epsilon, which prevents division by zero.
\end{remark}

\begin{remark}
All candidate radii use the same working matrix, boundary rule, initialization policy, numerical tolerance, and iteration limit. A warning is issued if one or more candidate fits fail to converge within the iteration
limit; the implemented selector nevertheless minimizes the criterion over all candidate fits. If the minimum is attained by more than one radius, the smallest is selected. Since the objective in~(\ref{ref6}) is non-convex
jointly in $(f,C)$, the algorithm may reach a local stationary solution. The analyses use one initialization per candidate, and no multiple-start procedure is applied.
\end{remark}
  
\subsection{Selection of the spatial radius}

The alternating least-squares procedure is applied to every candidate radius \(k\in\mathcal K=\{0,\ldots,s_{\max}\}\). Let \(\widehat Z_k\) denote the
resulting fitted coefficient matrix, and define
\[
  R_k=Z-\widehat Z_k,
  \qquad
  \widehat\sigma_k^2
  =
  \frac{\|R_k\|_F^2}{Nm},
  \qquad
  p_k=m(q_k+1).
\]
Here, $p_k$ counts the $q_k$ loading vectors and the intercept vector, each of dimension $m$, conditional on the fitted score field.

The implemented BIC-type reconstruction criterion is
\begin{equation}
  \operatorname{BIC}(k)
  =
  Nm\log\left\{
    \max\left(\widehat\sigma_k^2,\epsilon_{\mathrm{mach}}\right)
  \right\}
  +
  p_k\log(Nm).
  \label{eq:bic}
\end{equation}
Here, $\epsilon_{\mathrm{mach}}$ is the machine-level positive floor used by the implementation to avoid taking the logarithm of zero. The selected radius is $\widehat k\in\argmin_{k\in\mathcal K}\operatorname{BIC}(k)$,
with the smallest minimizer selected in the event of a tie.

This criterion is a conditional BIC-type reconstruction criterion rather than
a full likelihood-based BIC for the joint latent-field model. It treats the
$Nm$ entries of the coefficient matrix as scalar reconstruction targets and
does not specify a likelihood for their spatial or cross-coefficient dependence. It penalizes the loading and intercept parameters whose number increases with the Chebyshev radius.

The latent field $f$ is treated as a fitted component-score vector and is not included in $p_k$. Its nominal dimension is the same for every candidate radius, so adding a penalty based only on that nominal dimension would
contribute the same term for all $k$ and would not alter the selected radius. Nevertheless, because the latent field is re-estimated jointly with a radius-dependent loading matrix, the parameter count $p_k$ does not represent
a complete effective-degrees-of-freedom calculation for the joint nonlinear fit. The criterion should therefore be interpreted as a
reconstruction-based radius-selection rule, in the spirit of generalized dynamic principal components \citep{pena_yohai_gdpc}, rather than as a likelihood BIC for the full latent spatial model.
Although an approximate leave-one-out criterion and alternative penalized-reconstruction criteria are available in the implementation, all analyses in this paper use the BIC-type criterion in~\eqref{eq:bic}.

\subsection{Several components and reconstruction error}
\label{sec:reconstruction-error}

If $p\geq1$ components are requested, set $E^{(0)}=Z$. At step
$\ell=1,\ldots,p$, one component is fitted to $E^{(\ell-1)}$, with fitted
contribution
$\widehat E_{\mathrm{comp}}^{(\ell)}
=F_{\widehat k_\ell}(\widehat f^{(\ell)})
\widehat C_{\widehat k_\ell}^{(\ell)}\in\R^{N\times m}$,
and the residual is updated as
$E^{(\ell)}=E^{(\ell-1)}-\widehat E_{\mathrm{comp}}^{(\ell)}$.

The cumulative reconstruction after $p$ sequentially extracted components is
\[
  \widehat Z^{(p)}
  =
  \sum_{\ell=1}^{p}\widehat E_{\mathrm{comp}}^{(\ell)},
  \qquad
  E^{(p)}=Z-\widehat Z^{(p)}.
\]
This is sequential residual fitting rather than simultaneous optimization over all components, and the resulting components need not be mutually orthogonal. In the SST application, only one component is fitted, so no deflation step is required.

The cumulative coefficient-space mean squared error (MSE) and normalized mean squared error (NMSE) are
\[
  \operatorname{MSE}_p
  =
  \frac{\left\|Z-\widehat Z^{(p)}\right\|_F^2}{Nm},
  \qquad
  \operatorname{NMSE}_p
  =
  \frac{\left\|Z-\widehat Z^{(p)}\right\|_F^2}{\|Z\|_F^2}.
\]
The denominator $\|Z\|_F^2$ is computed once from the original input matrix and is not recomputed from the residual matrix at any deflation step.

For all cross-method comparisons below, the coefficient field was column-centred across spatial locations, separately for each basis coefficient, before fitting. The same centred field was used as the reference for all reconstruction errors. SFGDPC NMSE was computed using the unweighted coefficient-space Frobenius norm, whereas SFPCA and, where included, functional principal component analysis (FPCA) used functional \(L^2\)-NMSE based on the basis Gram matrix. The relationship between these conventions is given below.

For a non-orthonormal basis with Gram matrix $G$, the corresponding cumulative functional NMSE is
\[
  \operatorname{NMSE}_{L^2,p}
  =
  \frac{
    \displaystyle
    \sum_{s\in\cD}
    \left(z_s-\widehat z_s^{(p)}\right)^\top
    G
    \left(z_s-\widehat z_s^{(p)}\right)
  }{
    \displaystyle
    \sum_{s\in\cD} z_s^\top G z_s
  }.
\]
The 15-term Fourier bases used in both simulation studies and the SST application were numerically orthonormal. In each case, the maximum absolute
deviation of the evaluated Gram matrix from \(I_{15}\) was \(1.38\times10^{-6}\). Coefficient-space NMSE and Gram-weighted functional \(L^2\)-NMSE were numerically equivalent at the reported precision. They coincide exactly when the basis is orthonormal.

\section{Simulation Studies}
\label{sec:simulation}
We conducted two simulation experiments to evaluate SFGDPC under complementary forms of spatial functional dependence. The first used a stationary SFARMA-type process and served as a conventional benchmark for comparison with SFPCA. The second used a finite-support local-neighbourhood transfer mechanism, representing the setting for which neighbourhood-based reconstruction is specifically intended. Ordinary FPCA was included as a non-spatial reference.
In each replication, all methods were fitted to the same simulated coefficient field and retained the same number of components. The coefficient field was column-centred once across grid locations, separately for each basis coefficient, before being supplied to SFGDPC, SFPCA, and FPCA. Cumulative reconstruction error for each method was evaluated relative to this common centred field under the norm conventions described in
Section~\ref{sec:reconstruction-error}.
We report Monte Carlo means and standard deviations of cumulative NMSE. Radius recovery is examined separately for the first SFGDPC component when the data-generating mechanism has a finite neighbourhood radius.

\subsection{Simulation 1: Stationary SFARMA Benchmark}
The first experiment considered a stationary SFARMA-type spatial functional process. This setting provides a standard stationary dependence structure for which spectral spatial functional principal component methods are naturally suited \citep{kuenzer_sfpca}. The purpose of the experiment was to assess whether SFGDPC remained competitive when the data-generating mechanism was not specifically based on a finite local neighbourhood.

The experiment used 1000 Monte Carlo replications for each innovation setting: Gaussian innovations and Student's \(t\) innovations with five degrees of freedom (\(t_5\)). The heavier-tailed condition was included to assess sensitivity to departures from Gaussianity. In each replication, ordinary FPCA, SFPCA, and SFGDPC were fitted to the same simulated spatial functional field, and cumulative reconstruction NMSE was evaluated after one, two, and three retained components. SFGDPC selected its neighbourhood radius separately for each sequential component using the BIC-type reconstruction criterion.

Each simulated field was observed on a \(50\times50\) grid and represented using a 15-term Fourier basis. SFPCA used a supplied maximum filter radius of \(L_{\max}=3\) and a filter-mass threshold of \(\beta=1\), so no additional mass-based filter truncation was applied. Its spectral smoothing bandwidth \(q\) was chosen using the data-adaptive
rule described in Supplementary Section~S1. The spectral density was evaluated on a \(100\times100\) frequency grid. SFGDPC searched over \(k=0,\ldots,3\) using the BIC-type criterion, with convergence tolerance \(10^{-4}\) and a maximum of 500 iterations. All methods were fitted to the same centred field within each replication.
The complete SFARMA data-generating equations, parameter values, basis construction, and implementation settings appear in Supplementary Section~S1.

\begin{table}[ht]
\centering
\small
\setlength{\tabcolsep}{5pt}
\caption{Cumulative reconstruction NMSE under the stationary SFARMA benchmark over 1,000 Monte Carlo replications. Entries are Monte Carlo means with standard deviations in parentheses. Lower NMSE corresponds to better reconstruction; boldface identifies the lowest mean within each row and does not imply statistical significance.}
\label{tab:sfarma-1000-summary}
\begin{tabular}{llccc}
\toprule
Innovation
& Components
& FPCA
& SFPCA
& SFGDPC \\
\midrule
Gaussian
& 1
& 0.5990 (0.0427)
& \textbf{0.4697 (0.0394)}
& 0.4708 (0.0378) \\
Gaussian
& 2
& 0.3719 (0.0316)
& 0.2938 (0.0220)
& \textbf{0.2863 (0.0200)} \\
Gaussian
& 3
& 0.2357 (0.0218)
& 0.1950 (0.0155)
& \textbf{0.1852 (0.0119)} \\
\midrule
$t_5$
& 1
& 0.6019 (0.0425)
& \textbf{0.4751 (0.0396)}
& 0.4752 (0.0378) \\
$t_5$
& 2
& 0.3729 (0.0303)
& 0.2966 (0.0231)
& \textbf{0.2889 (0.0206)} \\
$t_5$
& 3
& 0.2365 (0.0208)
& 0.1965 (0.0169)
& \textbf{0.1862 (0.0121)} \\
\bottomrule
\end{tabular}
\end{table}

Table~\ref{tab:sfarma-1000-summary} shows similar one-component reconstruction performance for SFPCA and SFGDPC under both innovation distributions. After two and three retained components, SFGDPC had the lowest mean cumulative NMSE. The same pattern under Gaussian and \(t_5\) innovations indicates that this finding was not sensitive to the heavier-tailed condition examined.

The BIC-selected SFGDPC radius was strongly concentrated at \(k=1\), including all first-component fits. All selected fits met the implemented numerical stopping criterion, and cumulative NMSE was non-increasing across retained components in every replication. Since the SFARMA process has no unique generating Chebyshev radius, these selections are interpreted descriptively. Supplementary Section~S1 gives the detailed simulation settings
and component-wise summaries of the selected radii.

\subsection{Simulation 2: Local-Neighbourhood Transfer Model}
The second experiment evaluated the intended operating regime of SFGDPC. The coefficient vector at each spatial location was generated from values of a latent spatial field within a finite Chebyshev neighbourhood.

Let $\mathcal S=\{s_{uv}=(u,v):
u=1,\ldots,n_{\mathrm{lat}},\
v=1,\ldots,n_{\mathrm{lon}}\}$ denote the regular rectangular grid, and let
$\bm f=(f_1,\ldots,f_N)^\top$ be a scalar Gaussian spatial field. Here, \(n_{\mathrm{lat}}\) and \(n_{\mathrm{lon}}\) denote the numbers of grid points in the latitude and longitude directions, respectively, and
\(N=n_{\mathrm{lat}}n_{\mathrm{lon}}\). Under the nonstationary treatment, its covariance was generated using a nonstationary covariance construction \citep{paciorek_schervish_2006},
\[
C(s,s') =
\frac{2\rho_s\rho_{s'}}
{\rho_s^2+\rho_{s'}^2}
\exp\left\{
-\frac{\sqrt{2}\lVert s-s'\rVert}
{\sqrt{\rho_s^2+\rho_{s'}^2}}
\right\}.
\]
The local range increased linearly along the first grid direction from
$\rho_{\mathrm{low}}=1$ to $\rho_{\mathrm{high}}=4$, whereas the stationary
control used the constant range $\rho_0=2.5$. Each realization of the latent
field was centred and standardized before constructing the coefficient field.
The treatment and control conditions were paired by applying their respective
covariance transformations to the same replication-specific standard-normal
draws.

For a true neighbourhood radius $k_{\mathrm{true}}$, let
$\mathcal H_{k_{\mathrm{true}}}
=\{(a,b)\in\mathbb Z^2:|a|\vee|b|\le k_{\mathrm{true}}\}$
denote the full Chebyshev offset set, including the centre. For a target
location $s_i=(u_i,v_i)$ and offset $h=(a,b)$, define the clamping map
\[
\pi(s_i,h)
=
\left(
\min\{\max(u_i+a,1),n_{\mathrm{lat}}\},
\min\{\max(v_i+b,1),n_{\mathrm{lon}}\}
\right).
\]
The ordered vector of latent-field values in the clamped Chebyshev window
centred at $s_i$ is
$\bm{\phi}_i(\bm f)
=(f_{\pi(s_i,h)}:h\in\mathcal H_{k_{\mathrm{true}}})^\top$.
The coefficient vector was then generated as
\[
\bm z_i
=
B\bm\phi_i(\bm f)+\bm\varepsilon_i,
\qquad i=1,\ldots,N,
\]
where \(B\) was a deterministic non-separable loading matrix and, conditional on the generated signal, the entries of \(\bm\varepsilon_i\) were independent, mean-zero Gaussian variables with the coefficient-specific standard deviations specified below. The SFGDPC estimator retained the compacting self-padding convention described in the methodology section. Exact correspondence between the generating and fitted neighbourhood constructions is interpreted primarily for interior locations.

The loading matrix was fixed across replications and normalized so that
$\lVert B\rVert_F=\sqrt{m}$. Its explicit construction, including the offset-dependent smooth and directional terms, appears in Supplementary Section~S2.

The main experiment used $n_{\mathrm{lat}}=n_{\mathrm{lon}}=40$, $m=15$,
$k_{\mathrm{true}}=2$, and a noise ratio of $0.15$. The signal at each interior location depended directly on
$(2k_{\mathrm{true}}+1)^2=25$ latent-field values. For coefficient dimension $\ell$, the noise standard deviation was
$\sigma_{\varepsilon,\ell}
=0.15\,\operatorname{sd}(\mu_{1\ell},\ldots,\mu_{N\ell})$,
where $\bm{\mu}_i=B\bm{\phi}_i(\bm f)$.

The observed coefficient matrix was centred across spatial locations, separately for each basis coefficient, before fitting all methods. No column standardization was applied. SFGDPC used this centred matrix without further centring or rescaling. Cumulative reconstruction NMSE for SFGDPC, SFPCA,
and FPCA was evaluated relative to the same centred observed coefficient field.

The nonstationary treatment and stationary control were each evaluated using 100 Monte Carlo replications, and three components were retained. SFGDPC was fitted with \(s_{\max}=3\), using a convergence tolerance of \(10^{-4}\) and a maximum of 500 iterations. The neighbourhood radius was selected separately for each sequential component using the BIC-type reconstruction criterion.

SFPCA retained the same number of components and was fitted using the \texttt{fsd} R implementation \citep{kuenzer_sfpca},
available from the
\href{https://github.com/kuenzer/fsd}{GitHub repository}. The maximum supplied spatial filter radius was \(L_{\max}=9\), no filter-mass truncation was applied (\(\beta=1\)), and the spectral density was evaluated on a \(200\times200\) frequency grid. The \(q\leq1\) fallback was not triggered in any replication, and the
effective SFPCA filter radius was \(L=9\) in all 100 replications under both covariance conditions. The radius $L=9$ is the largest boundary-safe integer satisfying the interior-validity condition $1+2L\leq n-2L$ for a $40\times40$ grid. This boundary-safe condition ensures a nonempty interior region for the SFPCA interior error diagnostic; the reported NMSE values are whole-grid reconstruction errors, with no interior exclusion. Ordinary FPCA was again included as a non-spatial reference.

\begin{table}[H]
\centering
\small
\setlength{\tabcolsep}{6pt}
\caption{Cumulative reconstruction NMSE for the local-neighbourhood transfer simulation. Entries are Monte Carlo means with standard deviations in parentheses. Lower NMSE corresponds to better reconstruction; boldface identifies the lowest mean within each row and does not imply statistical significance.}
\label{tab:modelc-main-results}
\begin{tabular}{llccc}
\toprule
Condition
& \(p\)
& FPCA
& SFPCA
& SFGDPC \\
\midrule
Nonstationary
& 1
& 0.3634 (0.0397)
& 0.0819 (0.0087)
& \textbf{0.0506 (0.0051)} \\
Nonstationary
& 2
& 0.2109 (0.0268)
& 0.0561 (0.0071)
& \textbf{0.0334 (0.0035)} \\
Nonstationary
& 3
& 0.0982 (0.0127)
& 0.0497 (0.0068)
& \textbf{0.0247 (0.0027)} \\
\midrule
Stationary
& 1
& 0.3511 (0.0358)
& 0.0795 (0.0085)
& \textbf{0.0453 (0.0043)} \\
Stationary
& 2
& 0.1991 (0.0239)
& 0.0534 (0.0069)
& \textbf{0.0297 (0.0032)} \\
Stationary
& 3
& 0.0889 (0.0107)
& 0.0469 (0.0065)
& \textbf{0.0221 (0.0024)} \\
\bottomrule
\end{tabular}

\vspace{0.4em}
\begin{minipage}{0.92\textwidth}
\footnotesize
Notes: Results are based on 100 Monte Carlo replications on a \(40\times40\)
grid with \(k_{\mathrm{true}}=2\), \(m=15\), and noise ratio \(0.15\).
SFGDPC used BIC selection over \(k=0,\ldots,3\). SFPCA used
\(L_{\max}=9\), no filter-mass truncation (\(\beta=1\)), and a
\(200\times200\) frequency grid.
\end{minipage}
\end{table}

Table~\ref{tab:modelc-main-results} shows that SFGDPC had the lowest mean cumulative NMSE for every component count under both covariance conditions.
SFGDPC also attained lower NMSE than SFPCA in all 100 replications for each component count, while ordinary FPCA had the largest mean NMSE in every case.

The persistence of the reconstruction difference under the stationary control indicates that it cannot be attributed solely to the spatially varying covariance range. Instead, the result is consistent with an advantage from matching the reconstruction method to the finite-support local transfer mechanism used to generate the coefficient field.

In the main \(40\times40\) design, the first SFGDPC component selected the generating radius \(k_{\mathrm{true}}=2\) in all 100 replications under both covariance conditions. This is evidence of radius recovery under this design, rather than a claim of uniform recovery. Supplementary Section~S3.1 gives the subsequent-component selections.

The supplementary sensitivity study produced exact first-component recovery
under the baseline, \(60\times60\), higher-noise, and lag-zero-loading
(\(k_{\mathrm{true}}=0\)) conditions. The selector instead chose a smaller radius for every replication on the \(20\times20\) grid and under the generating-radius-three condition.
These results indicate that the selected radius represents the dominant reconstruction scale under the examined conditions but need not recover the complete generating support; see Supplementary Section~S3.2.

\subsubsection{Sensitivity to the SFPCA Filter Radius}

We assessed the sensitivity of the comparison to the SFPCA filter radius using
the centred, nonstationary local-neighbourhood transfer design. SFPCA was
evaluated at \(L\in\{3,5,7,9\}\) on the \(40\times40\) grid and at
\(L\in\{9,11,13,14\}\) on the \(60\times60\) grid. Within each replication,
the candidate radii were applied to the same simulated coefficient field,
while SFGDPC was fitted once because its estimates do not depend on \(L\).

Mean SFPCA cumulative NMSE decreased monotonically with \(L\) for all retained component counts. SFGDPC retained lower NMSE for \(p=1,2,3\) at the largest boundary-safe radius on each grid, with the same ordering observed in all 100 replications. This conclusion is restricted to the evaluated radii and the
implemented boundary conventions. Full numerical results are given in Supplementary Table~S4.

\subsubsection{Illustrative First-Component Spatial Filters}

Figure~\ref{fig:modelc-filter-recovery} shows the true local loading kernel and the first-component spatial filters estimated from one realization of the centred nonstationary treatment. The realization used the \(40\times40\) grid, \(k_{\mathrm{true}}=2\), and noise ratio \(0.15\), matching the main Monte Carlo design. Details of filter normalization and sign alignment are provided in Supplementary Section~S2.1.

\begin{figure}[htbp]
\centering
\includegraphics[width=\textwidth]
{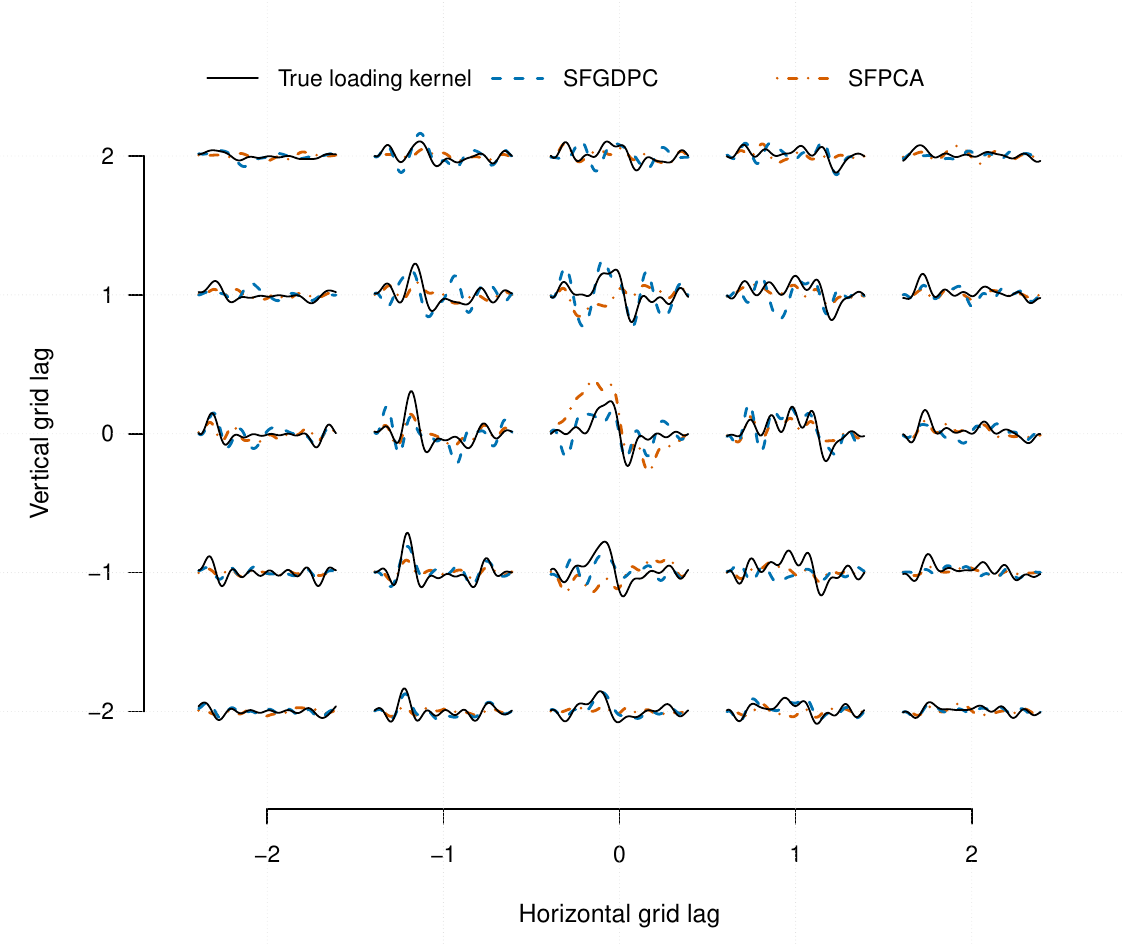}
\caption{True local loading kernel and estimated first-component
SFGDPC and SFPCA filters for one simulated realization of the
nonstationary local-neighbourhood model. SFGDPC selected
\(\widehat{k}_1=2\), while SFPCA used \(L=9\). Only spatial lags
in \([-2,2]^2\) are displayed. Each filter family was normalized
to unit integrated squared mass over its full support, and all
curves use a common amplitude scale.}
\label{fig:modelc-filter-recovery}
\end{figure}

For this realization, the SFGDPC filter is broadly consistent with several principal central and neighbouring-lag features of the true kernel. The displayed SFGDPC lag representation corresponds to interior grid locations; boundary neighbourhoods use the compacting self-padding convention described in Section~\ref{sec:methodology}. The figure is a qualitative illustration, while the replicated cumulative NMSE results in Table~\ref{tab:modelc-main-results} constitute the primary performance evidence.

\FloatBarrier

\section{Indian Ocean Sea Surface Temperature Application}
\label{sec:sst}

\subsection{Data description and exploratory diagnostics}

We illustrate the proposed method using Indian Ocean sea surface temperature (SST) observations from 1982 to 2014. The original data are spatio-temporal, but for the present analysis they are represented as spatially indexed annual temperature curves. For each retained grid location and year, the within-year
temperature profile was represented using a 15-term Fourier basis, following the SST representation
of~\citet{kuenzer_sfpca}. The Fourier basis provides a natural representation of the annual temperature cycle. Each annual field is treated as a collection of functional observations indexed by spatial location. Each calendar year was analysed as a separate spatial-functional field; both methods were fitted independently to all 33 annual fields, and temporal dependence between years was not modelled. The application is spatial functional rather than a joint spatio-temporal analysis.

The preprocessing follows the spatial functional SST setup used for the SFPCA comparison. At each spatial location, the mean annual curve over 1982--2014 was subtracted to form annual anomaly fields. The resulting coefficient array was then spatially coarsened by
retaining every third longitude and latitude grid point, beginning with the first grid point in each direction, while retaining all 33 annual fields. Each annual coefficient field was subsequently centred across the retained spatial locations before fitting SFGDPC; the SFPCA implementation performs the corresponding spatial centring internally. The resulting coefficient array had dimension \(15\times45\times40\times33\), corresponding to 15 Fourier coefficients, 45 longitude grid points, 40 latitude grid points, and 33 annual fields. Equivalently, each annual field was analysed on a \(40\times45\) latitude-by-longitude grid, giving
$
N=40\times45=1800
$
spatial locations. Hence, for each year, the data were represented by a coefficient matrix
$
Z\in\mathbb{R}^{1800\times15},
$
with rows corresponding to spatial grid locations and columns corresponding to Fourier coefficient dimensions.

The global overview in Figure~\ref{fig:indian-overview}(a) uses the National Oceanic and Atmospheric Administration (NOAA) Optimum Interpolation Sea Surface Temperature (OISST) product, which combines satellite and in situ observations on a regular global grid \citep{reynolds_etal_2007}. The study region exhibits a pronounced north--south temperature gradient, with warmer conditions in the north and cooler conditions in the south (Figure~\ref{fig:indian-overview}). The representative annual curves further reveal spatial differences in both average temperature and seasonal variation, motivating the treatment of the observations as spatially indexed functions.

\begin{figure}[!htbp]
\centering

\begin{minipage}{0.96\textwidth}
    \centering
    {\large\textbf{(a)}}\par
    \vspace{0.25em}
    \includegraphics[width=\textwidth]
    {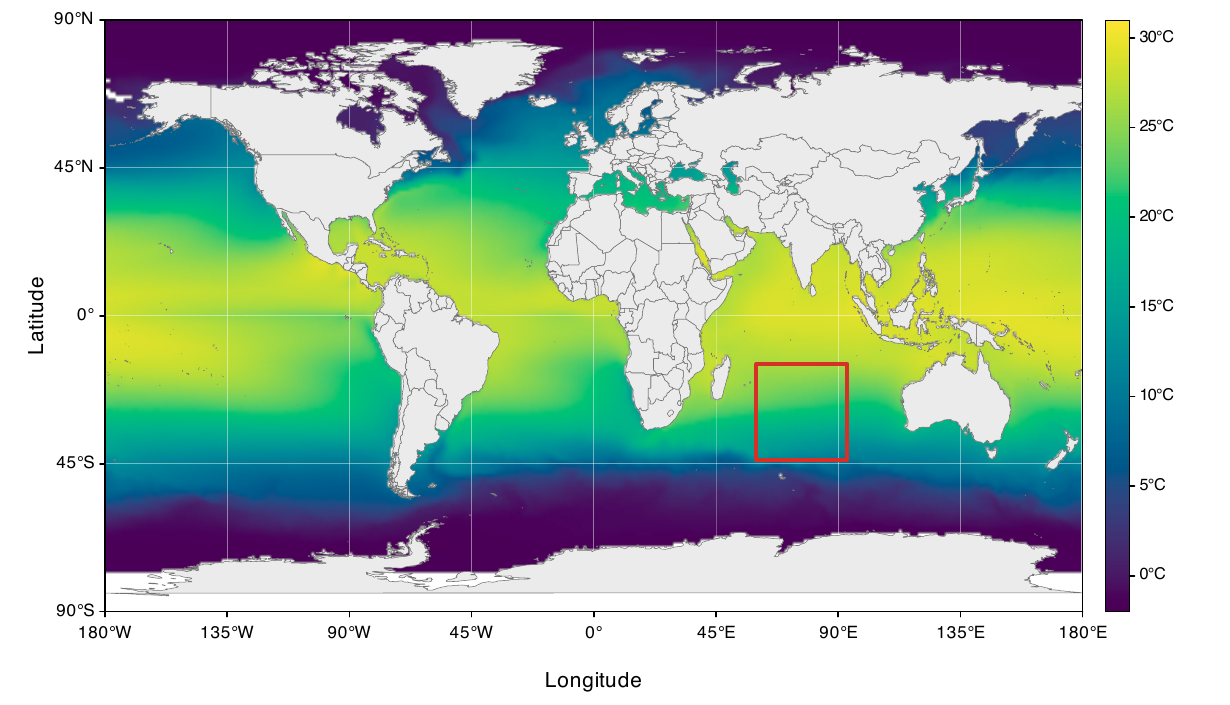}
\end{minipage}

\vspace{0.8em}

\begin{minipage}[t]{0.48\textwidth}
    \centering
    {\large\textbf{(b)}}\par
    \vspace{0.25em}
    \includegraphics[width=\textwidth]
    {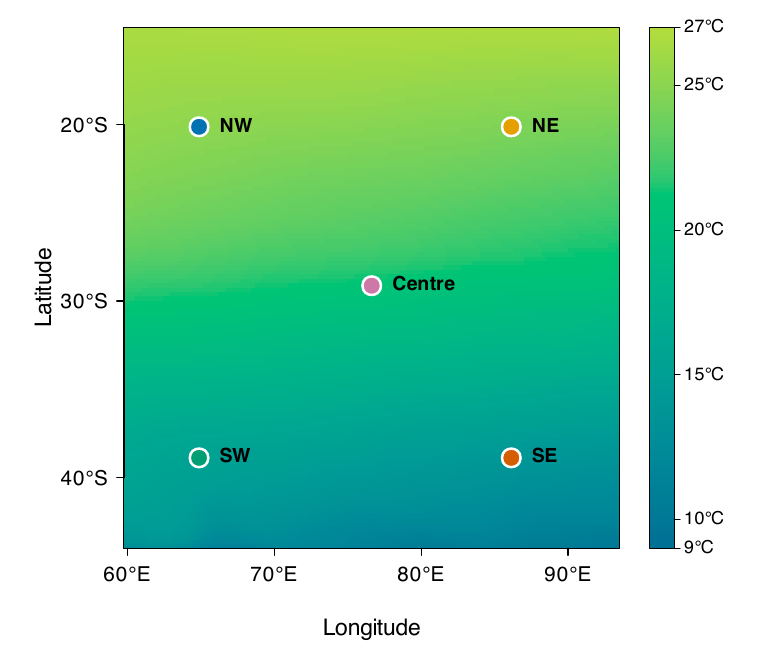}
\end{minipage}
\hfill
\begin{minipage}[t]{0.48\textwidth}
    \centering
    {\large\textbf{(c)}}\par
    \vspace{0.25em}
    \includegraphics[width=\textwidth]
    {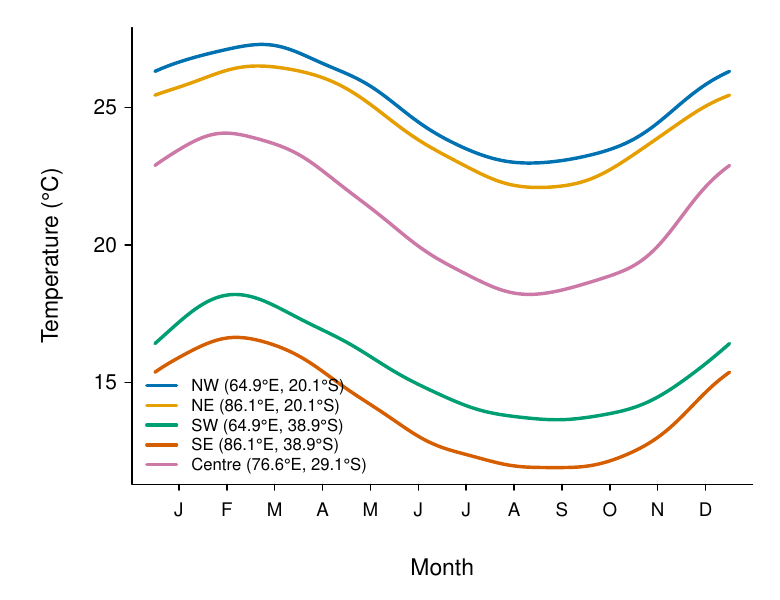}
\end{minipage}

\caption{Geographical overview and sea surface temperature (SST) summaries for the Indian Ocean application, 1982--2014. 
(a) Global long-term mean NOAA OISST, displayed at approximately \(1^\circ\) spacing from the \(0.25^\circ\) source grid; the red rectangle marks the study region. 
(b) Long-term mean SST calculated from the functional dataset on the native \(0.25^\circ\) grid; the coloured points mark the five selected locations. 
(c) Average annual SST curves at these locations.}
\label{fig:indian-overview}
\end{figure}

Figure~\ref{fig:indian-correlogram} presents an empirical functional spatial correlogram of the raw annual SST curves after centring across locations on the coarsened diagnostic grid. At each Chebyshev lag \(h\), normalized coefficient inner products were averaged over valid pairs at offsets \((h,0)\), \((0,h)\), \((h,h)\), and
\((h,-h)\), and then over years. This exploratory summary was computed without subtracting the location-specific mean annual
curves used to construct the anomaly fields analysed in Section~4.2. The plotted similarity decreases as the spatial lag increases, with stronger similarity among neighbouring functional
observations than among distant ones.

This exploratory pattern is consistent with local spatial dependence and motivates the use of a local-neighbourhood reconstruction method. It should not be interpreted as evidence that the SST field follows an exact finite-neighbourhood data-generating mechanism.

\begin{figure}
\centering
\includegraphics[width=0.72\textwidth]{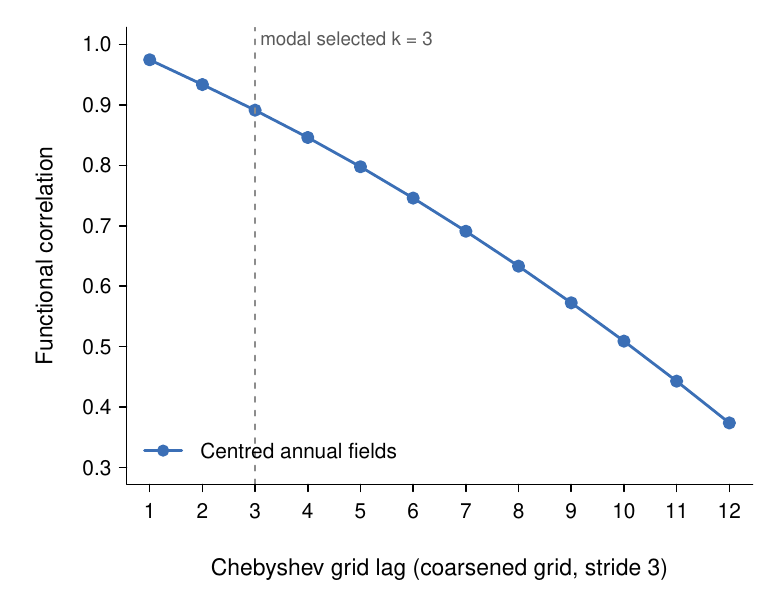}
\caption{Empirical functional spatial correlogram on the coarsened diagnostic grid. Functional similarity is computed from normalized
Fourier-coefficient inner products for raw annual SST curves after centring across locations. At each lag, values are averaged over valid pairs in four directions and then over years. The dashed vertical line marks the modal BIC-selected SFGDPC radius, \(k=3\).}
\label{fig:indian-correlogram}
\end{figure}

\FloatBarrier
\subsection{Reconstruction comparison}

For each annual coefficient field, we fitted one SFGDPC component and one SFPCA component. This is a one-component reconstruction comparison across 33 annual SST anomaly fields.

SFGDPC was fitted with a fixed maximum Chebyshev search radius
$
s_{\max}=4,
$
without additional column standardization of the coefficient matrix. The BIC-type reconstruction criterion selected the spatial radius \(k\) separately for each year. The convergence tolerance was \(10^{-4}\), and the maximum number of iterations was 500. The production driver also checked the fitted dimensions and verified equality between the stored NMSE and the manually computed NMSE. The selected fit for 2004 reached the 500-iteration limit and was
retained with the convergence flag set to \texttt{FALSE}; its attained reconstruction was used in the reported comparison. 

For comparison, SFPCA was fitted to the same spatially coarsened annual spatial-functional data with one component using the \texttt{fsd} R
implementation for SFPCA \citep{kuenzer_sfpca}, available from the
\href{https://github.com/kuenzer/fsd}{GitHub repository}. The main high-cap benchmark used a maximum supplied filter radius of \(L_{\max}=42\), a filter-mass threshold of \(\beta=0.95\), and a \(400\times400\) frequency grid for spectral-density estimation. Under this setting, the selected effective SFPCA filter radius did not reach the supplied cap; the largest selected value was \(L=41\).

We also report an SFPCA benchmark with \(L_{\max}=9\). Both SST benchmarks used \(\beta=0.95\), a \(400\times400\) frequency grid, and whole-grid reconstruction error without interior exclusion; only the supplied maximum filter radius differed.

Applying the interior-validity condition \(1+2L \le n-2L\) to the limiting grid dimension \(n=40\) gives \(L\le9.75\), so the largest boundary-safe integer radius is \(L=9\). The high-cap \(L_{\max}=42\) run is treated as the main empirical SFPCA benchmark
because it allows a wider candidate filter range under the whole-grid reconstruction criterion. 

\begin{table}
\centering
\caption{Summary of one-component whole-grid reconstruction performance for the
Indian Ocean SST anomaly fields, 1982--2014. Two SFPCA benchmarks are reported:
a boundary-safe reference setting with \(L_{\max}=9\), and a generous high-cap
setting with \(L_{\max}=42\), \(\beta=0.95\), and a \(400\times400\) frequency grid. SFGDPC was fitted with \(s_{\max}=4\) and BIC-selected \(k\).}
\label{tab:indian-summary-comparison}
\begin{tabular}{lccc}
\toprule
Method & Mean NMSE & Median NMSE & Years best \\
\midrule
SFPCA, \(L_{\max}=42\) & 0.465 & 0.479 & 0/33 \\
SFPCA, \(L_{\max}=9\)  & 0.518 & 0.530 & 0/33 \\
SFGDPC, \(s_{\max}=4\) & 0.284 & 0.283 & 33/33 \\
\bottomrule
\end{tabular}
\end{table}

Table~\ref{tab:indian-summary-comparison} shows that SFGDPC had lower one-component whole-grid NMSE than both SFPCA benchmarks in all 33 annual fields. Mean NMSE was \(0.284\) for SFGDPC, compared with \(0.465\) for the high-cap SFPCA benchmark and \(0.518\) for the boundary-safe benchmark.

\begin{figure}
\centering
\includegraphics[width=0.85\textwidth]{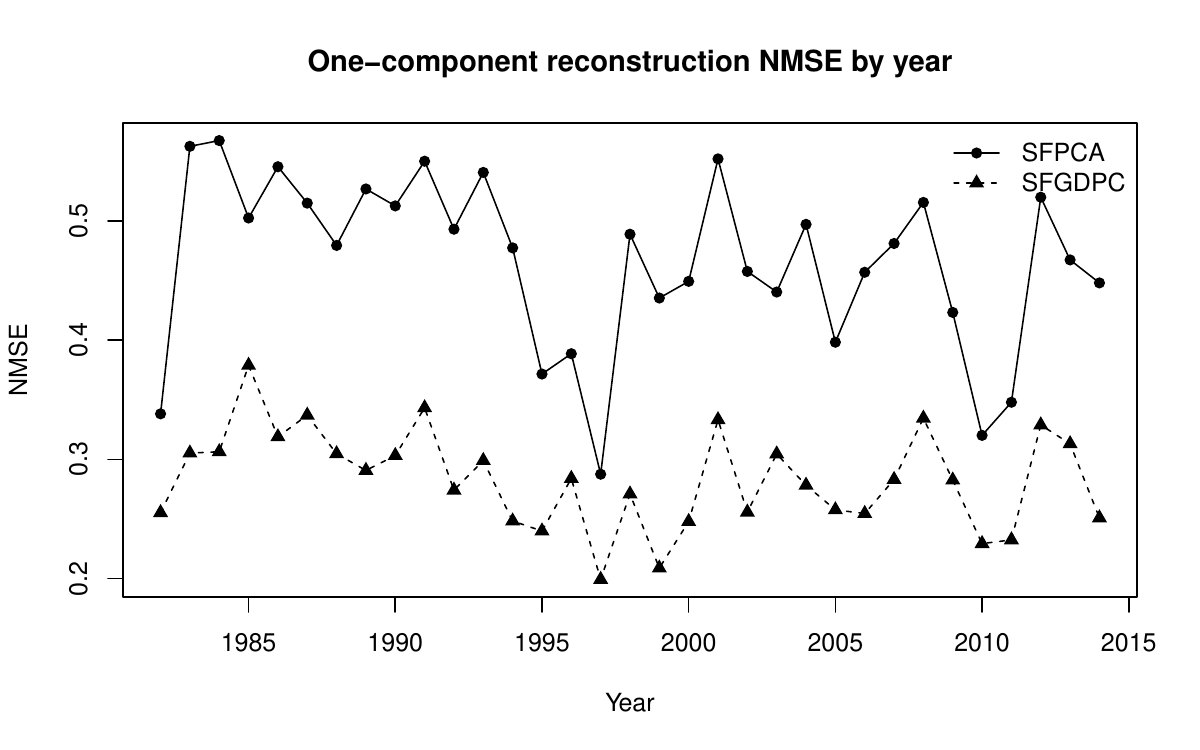}
\caption{One-component whole-grid reconstruction NMSE by year for the Indian Ocean SST anomaly fields. The SFPCA curve corresponds to the main high-cap benchmark with \(L_{\max}=42\), \(\beta=0.95\), and a \(400\times400\) frequency grid. Both methods were fitted to the same coarsened annual coefficient fields.}
\label{fig:indian-nmse-by-year}
\end{figure}

Figure~\ref{fig:indian-nmse-by-year} displays the year-by-year one-component NMSE values. The SFGDPC curve lies below the SFPCA curve throughout the study period, indicating that the lower average error is not driven by a single anomalous year.
Figure~\ref{fig:indian-selected-k} shows the BIC-selected SFGDPC spatial radius for each year. The selected radii were concentrated at small local
neighbourhoods. The BIC-selected radius was \(k=3\) in 25 of the 33 years, \(k=2\) in 6 years, and \(k=4\) in 2 years. The fitted leading component usually used a \(7\times7\) Chebyshev window. The two years with \(k=4\) reach the imposed search cap \(s_{\max}=4\). Overall, the selected radii indicate predominantly local reconstruction scales within the candidate set, most often \(k=2\) and \(k=3\), while leaving open the possibility that a
larger search cap could be relevant for a small number of years.

\begin{figure}[ht]
\centering
\includegraphics[width=0.85\textwidth]{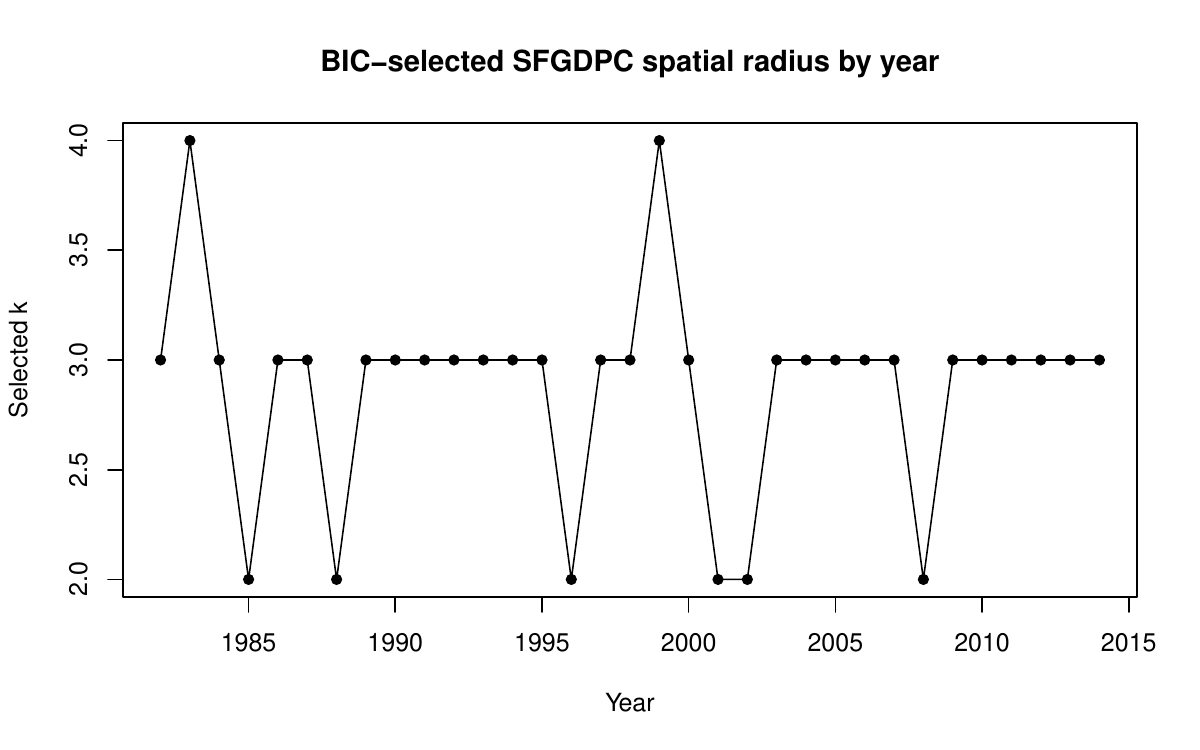}
\caption{BIC-selected SFGDPC spatial radius by year for the Indian Ocean SST anomaly fields.}
\label{fig:indian-selected-k}
\end{figure}

Neither the annual reconstruction errors nor the selected radii exhibited a clear monotonic temporal pattern.

\FloatBarrier

\section{Discussion}
\label{sec:discussion}

SFGDPC builds on the finite-sample reconstruction perspective underlying
GDPC \citep{pena_yohai_gdpc}, which was developed for vector time series
with the aim of adapting to possible nonstationary behaviour. In the spatial functional setting, SFPCA \citep{kuenzer_sfpca} provides a spectral principal-component decomposition for weakly stationary functional random fields on regular grids. SFGDPC extends the reconstruction principle of GDPC
to functional observations on a two-dimensional spatial grid and selects the neighbourhood radius as part of each component fit. This reconstruction-based approach complements spectral SFPCA. Its finite-sample criterion does not require second-order stationarity, isotropy, or estimation of a spatial spectral density.

The contrasting simulation settings clarify the role of the proposed method. In the stationary SFARMA benchmark, mean one-component NMSE was nearly identical for the two spatial methods, with SFPCA marginally lower. SFGDPC gave lower mean cumulative NMSE after two and three components were retained. Its advantage is stronger in the local-neighbourhood transfer experiment and persists under both the stationary control and the nonstationary treatment.

The persistence of the reconstruction advantage under the stationary control shows that it does not depend on the nonstationary covariance specification. The result is consistent with the suitability of neighbourhood-based
reconstruction for the finite-support local transfer mechanism used in this experiment. The SFPCA-radius sensitivity analysis showed lower whole-grid
cumulative NMSE for SFGDPC at every radius examined, under the implemented boundary conventions.

The selected SFGDPC radius provides an interpretable estimate of the dominant
spatial reconstruction scale. Exact recovery of the generating radius in the main simulation supports this interpretation, while the sensitivity results show that exact radius recovery is not uniform across all examined designs. In the SST application, SFGDPC had lower one-component whole-grid reconstruction error in all 33 annual fields. Within the specified search range \(k=0,\ldots,4\), the selected radii were concentrated at \(k=2\) and
\(k=3\), with the upper bound \(k=4\) selected in only two years. Together,
the stationary SFARMA benchmark and the SST application indicate that the
usefulness of SFGDPC is not confined to data generated from its own
reconstruction mechanism.

SFGDPC is particularly appropriate when the objective is accurate whole-field reconstruction through a compact neighbourhood, when the neighbourhood scale is of interest, or when stationarity is questionable. SFPCA remains a natural alternative when the scientific objective is a frequency-domain description of spatial functional dependence under weak stationarity. The choice between the methods should therefore reflect the
scientific objective and anticipated dependence structure.

The finite-sample proposition establishes existence of a minimizer for the fixed-radius problem restricted to $\cF_k$, and uniqueness of the conditional loading update when the spatial design has full column rank. The implemented Moore--Penrose inverse fallback also provides a computational update in rank-deficient cases. The current implementation is designed for complete regular rectangular grids and uses a dense spatial-field update, which may limit scalability to larger grids.

\section{Conclusion}
\label{sec:conclusion}

This paper introduced SFGDPC as a neighbourhood-based reconstruction method for spatial functional data on regular grids. For each candidate radius, SFGDPC alternates between estimation of the latent spatial field and the functional loadings, after which the component-specific radius is selected using the BIC-type reconstruction criterion. This provides a local and interpretable representation without requiring a stationary covariance model in the finite-sample fitting criterion. These results support its use when spatial information is concentrated in local neighbourhoods, while its performance in the stationary benchmark indicates that it remains competitive outside that setting.

The framework forms a basis for further development of reconstruction-based spatial functional dimension reduction. Further work may develop sparse or structured field-update solvers for larger grids, extend the neighbourhood construction to irregular spatial domains, and further evaluate data-driven selection of the number of components.

\singlespacing

\noindent{\large\bfseries Funding\par}
\vspace{0.5em}

This work was supported by the Universiti Malaya Research Excellence Grant,
Universiti Malaya [grant number UMREG044-2024].

\vspace{1.5em}

\bibliographystyle{elsarticle-num-names}
\bibliography{references_BDR}

@article{pena_yohai_gdpc,
  author  = {Pe{\~n}a, Daniel and Yohai, Victor J.},
  title   = {Generalized {D}ynamic {P}rincipal {C}omponents},
  journal = {J. Am. Stat. Assoc.},
  volume  = {111},
  number  = {515},
  pages   = {1121--1131},
  year    = {2016},
  doi     = {10.1080/01621459.2015.1072542}
}

@article{kuenzer_sfpca,
  author  = {Kuenzer, Thomas and H{\"o}rmann, Siegfried and Kokoszka, Piotr},
  title   = {Principal Component Analysis of Spatially Indexed Functions},
  journal = {J. Am. Stat. Assoc.},
  year    = {2021},
  volume  = {116},
  number  = {535},
  pages   = {1444--1456},
  doi     = {10.1080/01621459.2020.1732395}
}

@article{paciorek_schervish_2006,
  author  = {Paciorek, Christopher J. and Schervish, Mark J.},
  title   = {Spatial Modelling Using a New Class of Nonstationary Covariance Functions},
  journal = {Environmetrics},
  year    = {2006},
  volume  = {17},
  number  = {5},
  pages   = {483--506},
  doi     = {10.1002/env.785}
}

@article{delicado2010spatial,
  author  = {Delicado, Pedro and Giraldo, Ram{\'o}n and Comas, Carles and Mateu, Jorge},
  title   = {Statistics for Spatial Functional Data: Some Recent Contributions},
  journal = {Environmetrics},
  year    = {2010},
  volume  = {21},
  number  = {3--4},
  pages   = {224--239},
  doi     = {10.1002/env.1003}
}

@book{mateu2021geostatistical,
  editor    = {Mateu, Jorge and Giraldo, Ram{\'o}n},
  title     = {Geostatistical Functional Data Analysis},
  publisher = {John Wiley \& Sons},
  address   = {Hoboken, NJ},
  year      = {2021},
  doi       = {10.1002/9781119387916},
  isbn      = {9781119387848}
}

@article{reynolds_etal_2007,
  author  = {Reynolds, Richard W. and Smith, Thomas M. and Liu, Chunying and Chelton, Dudley B. and Casey, Kenneth S. and Schlax, Michael G.},
  title   = {Daily High-Resolution-Blended Analyses for Sea Surface Temperature},
  journal = {J. Clim.},
  year    = {2007},
  volume  = {20},
  number  = {22},
  pages   = {5473--5496},
  doi     = {10.1175/2007JCLI1824.1}
}

@article{calculli_etal_2015,
  author  = {Calculli, Crescenza and Fass{\`o}, Alessandro and Finazzi, Francesco and Pollice, Alessio and Turnone, Annarita},
  title   = {Maximum Likelihood Estimation of the Multivariate Hidden Dynamic Geostatistical Model with Application to Air Quality in {Apulia, Italy}},
  journal = {Environmetrics},
  year    = {2015},
  volume  = {26},
  number  = {6},
  pages   = {406--417},
  doi     = {10.1002/env.2345}
}

@article{wang_finazzi_fasso_2021,
  author  = {Wang, Yaqiong and Finazzi, Francesco and Fass{\`o}, Alessandro},
  title   = {{D-STEM} v2: A Software for Modeling Functional Spatio-Temporal Data},
  journal = {J. Stat. Softw.},
  year    = {2021},
  volume  = {99},
  number  = {10},
  pages   = {1--29},
  doi     = {10.18637/jss.v099.i10}
}

@article{maranzano_otto_fasso_2023,
  author  = {Maranzano, Paolo and Otto, Philipp and Fass{\`o}, Alessandro},
  title   = {Adaptive {LASSO} Estimation for Functional Hidden Dynamic Geostatistical Models},
  journal = {Stoch. Environ. Res. Risk Assess.},
  year    = {2023},
  volume  = {37},
  number  = {9},
  pages   = {3615--3637},
  doi     = {10.1007/s00477-023-02466-5}
}

@article{pena_smucler_yohai_2020,
  author  = {Pe{\~n}a, Daniel and Smucler, Ezequiel and Yohai, Victor J.},
  title   = {gdpc: An {R} Package for Generalized Dynamic Principal Components},
  journal = {J. Stat. Softw.},
  year    = {2020},
  volume  = {92},
  number  = {2},
  pages   = {1--23},
  doi     = {10.18637/jss.v092.c02}
}

@article{siahmed2025smfpca,
  author  = {Si-ahmed, Idris and Hamdad, Leila and Agonkoui, Christelle Judith and Kande, Yoba and Dabo-Niang, Sophie},
  title   = {Principal Component Analysis of Multivariate Spatial Functional Data},
  journal = {Big Data Res.},
  volume  = {39},
  pages   = {100504},
  year    = {2025},
  doi     = {10.1016/j.bdr.2024.100504}
}

@article{MenafoglioSecchiDallaRosa2013,
  author  = {Menafoglio, Alessandra and Secchi, Piercesare and Dalla Rosa, Matilde},
  title   = {A Universal Kriging Predictor for Spatially Dependent Functional Data of a {Hilbert} Space},
  journal = {Electron. J. Stat.},
  year    = {2013},
  volume  = {7},
  pages   = {2209--2240},
  doi     = {10.1214/13-EJS843}
}

@article{liu2012remote,
  author  = {Liu, Chong and Ray, Surajit and Hooker, Giles and Friedl, Mark},
  title   = {Functional Factor Analysis for Periodic Remote Sensing Data},
  journal = {Ann. Appl. Stat.},
  year    = {2012},
  volume  = {6},
  number  = {2},
  pages   = {601--624},
  doi     = {10.1214/11-AOAS518}
}

\end{document}


\begin{center}
{\large\bfseries Supplementary Material}\\[0.2em]
{\normalsize\bfseries
Neighbourhood-Based Generalized Dynamic Principal Components\\
for Spatial Functional Data}
\end{center}
\vspace{0.3em}

\setstretch{1.2}

\textbf{Notation.} Throughout this supplement, \(d\) denotes the number of basis coefficients per functional observation, written \(m\) in the main article; here \(d=m=15\) in all experiments.

\section{Stationary SFARMA Benchmark}
\label{sec:supp-sfarma}

This section provides the complete data-generating mechanism and implementation settings for the stationary SFARMA benchmark reported in the simulation section of the main article.

Let \(\bm{X}_{s,t}\in\mathbb{R}^{d}\), \(s,t=1,\ldots,50\), denote the basis-coefficient vector of the functional observation at grid location \((s,t)\). The coefficient field was generated from the spatial functional autoregressive moving-average model
\[
\bm{X}_{s,t}
=
A_{10}\bm{X}_{s-1,t}
+
A_{01}\bm{X}_{s,t-1}
+
\bm{\varepsilon}_{s,t}
+
B_{10}\bm{\varepsilon}_{s-1,t}
+
B_{01}\bm{\varepsilon}_{s,t-1}
+
B_{11}\bm{\varepsilon}_{s-1,t-1}.
\]
Thus, the autoregressive lag set was \(\mathcal{P}=\{(1,0),(0,1)\}\), and the moving-average lag set was
\(\mathcal{Q}=\{(0,0),(1,0),(0,1),(1,1)\}\).

The functional observations were represented using \(d=15\) Fourier basis functions on \([0,1]\). The innovation scale matrix was diagonal, with \(\Sigma_{\ell\ell}=\exp(-\ell/2)\) for
\(\ell=1,\ldots,d\), and zero off-diagonal entries.

The autoregressive and moving-average operators were regenerated in each Monte Carlo replication. For an operator associated with spatial lag \(\bm{h}\), an initial \(d\times d\) random matrix was generated as
\[
\widetilde{M}_{\bm{h},ab}
=
\frac{\xi_{\bm{h},ab}}
{\sqrt{a^2+b^2}},
\qquad
a,b=1,\ldots,d,
\]
where the \(\xi_{\bm{h},ab}\) were independent standard-normal random variables.

For each nonzero moving-average lag, the operator was normalized as
\[
B_{\bm{h}}
=
0.5\,
\frac{\widetilde{M}_{\bm{h}}}
{\lVert\widetilde{M}_{\bm{h}}\rVert_F},
\qquad
\bm{h}\in
\{(1,0),(0,1),(1,1)\},
\]
whereas the zero-lag moving-average operator was fixed as \(B_{00}=I_d\).

For the two autoregressive lags, first define
\[
\overline{A}_{\bm{h}}
=
\frac{\widetilde{M}_{\bm{h}}}
{\lVert\widetilde{M}_{\bm{h}}\rVert_F},
\qquad
\bm{h}\in\mathcal{P}.
\]
The common scaling constant was
\[
c_A
=
\frac{\sqrt{0.2}}
{
\sqrt{
\displaystyle
\max_{\bm{h},\bm{g}\in\mathcal{P}}
\left\|
\overline{A}_{\bm{h}}
\overline{A}_{\bm{g}}
\right\|_F
}
},
\]
and the final autoregressive operators were
\(A_{\bm h}=c_A\overline A_{\bm h}\), \(\bm h\in\mathcal P\). This scaling was used to produce a stable stationary SFARMA field.

The fields were generated using \texttt{fsd::fsd.sfarma} on a \(50\times50\) grid with a burn-in of 10, with \texttt{do.fixed.point = FALSE} and \texttt{max.iter = -1}. No additional boundary correction was applied outside the simulation convention implemented by \texttt{fsd.sfarma}.

Two innovation distributions were considered. The Gaussian condition used \texttt{noise = "normal"}, whereas the heavier-tailed condition used \texttt{noise = 5}, corresponding to the \(t_5\) innovation
option in \texttt{fsd.sfarma}.

 The Gaussian and \(t_5\) experiments used reproducible parallel random-number streams with initial seeds 147 and 10147, respectively. Both experiments used three parallel workers with L'Ecuyer--CMRG streams initialized by \texttt{clusterSetRNGStream} and static task scheduling. Exact reproduction of these runs requires the same worker count and scheduling. Within each replication, the simulated coefficient field was column-centred across the 2500 spatial locations, separately for each of the \(d=15\) basis coefficients. FPCA, SFPCA, and SFGDPC were then applied to the same centred field.

All methods retained three components. SFPCA selected its spectral smoothing parameter \(q\) using the data-adaptive rule implemented in \texttt{spca.automatic}. 
It was fitted with \(L_{\max}=3\), \(\beta=1\), and \(\texttt{freq.res}=50\). For a positive value of \texttt{freq.res}, each frequency axis contains twice that number of grid points. Thus, \texttt{freq.res = 50} gives a \(100\times100\) frequency grid. The setting \texttt{freq.res = 0} uses only the zero frequency. With \(\beta=1\), no additional filter-mass truncation was applied, so the effective filter radius was \(L=3\) by construction.

Ordinary FPCA was obtained through the same implementation by setting \(q=0\), \(L_{\max}=0\), and \(\texttt{freq.res}=0\). For the FPCA call, \texttt{inner} was set equal to the effective SFPCA radius. The effective radius was \(L=3\) in every replication, so \(\texttt{inner}=3\) was used throughout. The \texttt{inner} argument controls an additional interior-error diagnostic; the reported NMSE uses the whole grid.

SFGDPC was fitted with \(s_{\max}=3\), \texttt{crit = "BIC"}, \texttt{normalize = 1}, a convergence tolerance of \(10^{-4}\), and a maximum of 500 alternating least-squares iterations. No nondefault initialization or number of starts was supplied in the simulation driver, so the package defaults were used. The setting \texttt{normalize = 1} applies no additional centring
or scaling to the supplied coefficient matrix. The latent score field is still centred and standardized during alternating least-squares fitting.

For all three methods, performance after \(p\) components was evaluated using the cumulative reconstruction. In particular, the SFGDPC error
was computed as
\[
\operatorname{NMSE}_{\mathrm{SFGDPC}}(p)
=
\frac{
\left\|
Z-\displaystyle\sum_{j=1}^{p}\widehat{Z}_j
\right\|_F^2
}{
\lVert Z\rVert_F^2
},
\qquad
p=1,2,3,
\]
so that its definition agreed with the cumulative SFPCA reconstruction error.

The final experiment comprised 1000 Monte Carlo replications for each innovation distribution.

The SFARMA benchmark also provides an off-model assessment of the neighbourhood-radius selector. Although the model is defined using finite autoregressive and moving-average lag sets, the resulting autoregressive spatial process is not generated through a single finite Chebyshev loading
neighbourhood. Consequently, there is no unique
\(k_{\mathrm{true}}\) against which exact, under-, or over-selection can be defined. The selected radii in this benchmark are therefore reported descriptively rather than interpreted as recovery frequencies.

Table~\ref{tab:supp-sfarma-k} summarizes the neighbourhood radii selected by SFGDPC. The selected radius was concentrated at \(k=1\) under both innovation distributions.

\begin{table}[ht]
\centering
\small
\caption{Mean and standard deviation of the BIC-selected SFGDPC radius
under the stationary SFARMA benchmark over 1000 Monte Carlo
replications.}
\label{tab:supp-sfarma-k}
\begin{tabular}{lccc}
\toprule
Innovation & Component & Mean selected \(k\) & SD selected \(k\) \\
\midrule
Gaussian & 1 & 1.000 & 0.000 \\
Gaussian & 2 & 1.004 & 0.063 \\
Gaussian & 3 & 1.002 & 0.045 \\
\midrule
\(t_5\) & 1 & 1.000 & 0.000 \\
\(t_5\) & 2 & 1.004 & 0.063 \\
\(t_5\) & 3 & 1.007 & 0.083 \\
\bottomrule
\end{tabular}
\end{table}

\section{Construction of the Local Loading Matrix}
\label{sec:supp-loading-matrix}

For the local-neighbourhood transfer simulation, let
\[
\mathcal{D}_{k_{\mathrm{true}}}
=
\left\{
(a,b):
-k_{\mathrm{true}}
\leq a,b
\leq k_{\mathrm{true}}
\right\}
\]
denote the ordered set of spatial offsets. With \(k_{\mathrm{true}}=2\), the number of offsets was
\(M_{\mathrm{true}}=(2k_{\mathrm{true}}+1)^2=25\).

The offsets were generated using \texttt{expand.grid(di = -2:2, dj = -2:2)}. Consequently, the first coordinate varied fastest and the second coordinate varied slowest. In particular, the ordering began as
\((-2,-2),(-1,-2),(0,-2),(1,-2),(2,-2),(-2,-1),\ldots\).\\
Let \((a_r,b_r)\), \(r=1,\ldots,M_{\mathrm{true}}\), denote the
\(r\)th offset in this ordering. The loading matrix
\(B=[\bm b_1,\ldots,\bm b_{M_{\mathrm{true}}}]\), of dimension
\(d\times M_{\mathrm{true}}\), assigns a \(d\)-dimensional coefficient
profile to every spatial offset.

Define the coefficient-domain evaluation points
\(t_\ell=(\ell-1)/(d-1)\), \(\ell=1,\ldots,d\), with \(d=15\).
For offset \(r\), the spatial attenuation weight was \(w_r=\exp\{-0.20(a_r^2+b_r^2)\}\).

The smooth offset-specific profile was
\[
p_r(t)
=
\sin\left\{
\frac{(r+1)\pi t}{2}
\right\}
\exp(-0.8t)
+
0.4
\cos\left\{
\frac{(r+2)\pi t}{3}
\right\}
\exp(-0.5t),
\]
where \(r=1,\ldots,25\) follows the one-based indexing used in the
simulation code.

The directional contribution was
\(d_r(t)=0.25a_r\sin(2\pi t)+0.25b_r\cos(2\pi t)\).

Before normalization, the \(\ell\)th entry of the \(r\)th loading profile was therefore
\[
\widetilde{B}_{\ell r}
=
w_r
\left\{
p_r(t_\ell)
+
d_r(t_\ell)
\right\},
\qquad
\ell=1,\ldots,d,
\quad
r=1,\ldots,25.
\]

The complete loading matrix was finally normalized according to
\[
B=\sqrt{d}\,
\frac{\widetilde B}{\lVert\widetilde B\rVert_F},
\qquad
\lVert B\rVert_F=\sqrt d.
\]
The same deterministic matrix \(B\) was used in every Monte Carlo replication.

For a target grid location \((i,j)\), the latent-field value associated with offset \((a_r,b_r)\) was obtained from
\[
\widetilde{i}_{r}
=
\min\{
\max(i+a_r,1),
n_{\mathrm{lat}}
\},
\qquad
\widetilde{j}_{r}
=
\min\{
\max(j+b_r,1),
n_{\mathrm{lon}}
\}.
\]
Thus, coordinates extending beyond the observed grid were clamped to the nearest observed boundary coordinate. The noiseless coefficient vector was then
\[
\bm{\mu}_{ij}
=
B
\left(
f_{\widetilde{i}_1,\widetilde{j}_1},
\ldots,
f_{\widetilde{i}_{25},\widetilde{j}_{25}}
\right)^\top.
\]

For interior locations, this construction coincides with a complete Chebyshev neighbourhood of radius \(k_{\mathrm{true}}=2\). At boundary locations, the generator used nearest-boundary clamping, whereas the fitted SFGDPC method used the compacting self-padding
convention described in the main article. The generating and fitted boundary conventions were therefore not identical, and exact correspondence with the generating mechanism is interpreted primarily for interior locations.

\subsection{Construction of the illustrative filter comparison}
\label{sec:supp-filter-comparison}

The illustrative filter comparison in the main article was constructed from one simulated realization of the nonstationary local-neighbourhood model described in Section~3.2 of the main article. The displayed true
local loading kernel was formed directly from the 25 columns of the matrix \(B\), with each column assigned to its corresponding offset in \(\mathcal{D}_{2}\). It was not obtained from a separate principal component decomposition of \(B\).

Let \(G_{ab}=\int_0^1\psi_a(t)\psi_b(t)\,\mathrm{d}t\) denote the Gram matrix of the Fourier basis. For a collection of coefficient vectors \(C=[\bm c_1,\ldots,\bm c_J]\), the functional \(L^2\) mass used in the plotting script was
\[
\mathcal{M}(C)
=
\operatorname{tr}
\left(
C^\top G C
\right)
=
\sum_{j=1}^{J}
\bm{c}_j^\top G\bm{c}_j.
\]

The true kernel was normalized over its complete support
\(\mathcal{D}_2\). Each SFGDPC component was normalized over its selected compact neighbourhood support. Each SFPCA filter was first normalized using its complete support
\(\{\bm\ell:\lVert\bm\ell\rVert_\infty\leq L\}\), with \(L=9\),
and was then restricted to the displayed window
\(\mathcal W=\{\bm\ell:\lVert\bm\ell\rVert_\infty\leq2\}\).

The proportion of the SFPCA filter mass contained in the displayed window was calculated as
\(P_{\mathcal W}=\mathcal M(C_{\mathcal W})/
\mathcal M(C_{\mathrm{full}})\).

The signs of estimated components and filters are arbitrary. The first SFGDPC and SFPCA filter families were therefore oriented relative to the true lag-zero loading profile. Specifically, an estimated filter family was multiplied by \(-1\) whenever \(\bm c_0^\top G\bm b_0<0\), where \(\bm c_0\) is the estimated lag-zero coefficient vector and \(\bm b_0\) is the true lag-zero loading vector.

For the displayed SFPCA first filter, the reported window-mass proportion was \(0.765\).

\section{BIC Selection of the Neighbourhood Radius}
\label{sec:supp-radius-results}

\subsection{Selections across sequential components}
\label{sec:supp-sequential-k}

Table~\ref{tab:supp-modelc-k} reports the radii selected for all three sequential SFGDPC components in the main local-neighbourhood simulation. Recovery of the generating radius is interpreted primarily for the first component. Components 2 and 3 were fitted to sequential residual fields; their selected radii therefore describe the spatial scales remaining after removal of the preceding components. 

\begin{table}[ht]
\centering
\small
\caption{BIC-selected SFGDPC radius in the local-neighbourhood transfer
simulation over 100 Monte Carlo replications. For Components 2 and 3,
selection of \(k=2\) is descriptive because the components were fitted to
sequential residual fields.}
\label{tab:supp-modelc-k}
\begin{tabular}{llccc}
\toprule
\shortstack[l]{Covariance\\condition}
& Component
& \shortstack{Mean selected\\\(k\)}
& \shortstack{SD of selected\\\(k\)}
& \shortstack{Proportion selecting\\\(k=2\)} \\
\midrule
Nonstationary & 1 & 2.00 & 0.000 & 1.00 \\
Nonstationary & 2 & 1.19 & 0.394 & 0.19 \\
Nonstationary & 3 & 1.10 & 0.302 & 0.10 \\
\midrule
Stationary control & 1 & 2.00 & 0.000 & 1.00 \\
Stationary control & 2 & 1.25 & 0.435 & 0.25 \\
Stationary control & 3 & 1.07 & 0.256 & 0.07 \\
\bottomrule
\end{tabular}
\end{table}

\subsection{Sensitivity to field size, noise and generating radius}
\label{sec:supp-bic-sensitivity}

A focused sensitivity study examined selection of the first SFGDPC neighbourhood radius under changes in spatial field size, coefficient-wise noise level, and generating radius. Only the first component was examined
because \(k_{\mathrm{true}}\) defines the dominant generating neighbourhood before sequential residualization. This study concerns the behaviour of
the SFGDPC BIC selector and is separate from the reconstruction comparison between SFGDPC and SFPCA.

In this subsection, \(k_{\mathrm{gen}}\) denotes the nominal radius used to construct the finite loading neighbourhood in the data-generating mechanism; this quantity is recorded as \texttt{k\_true} in the simulation code. For \(k_{\mathrm{gen}}=0\), only the focal latent-field value enters the loading, although the latent Gaussian process itself remains spatially correlated. For the smaller grid, \(k_{\mathrm{gen}}=2\) applies exactly to interior
locations, whereas boundary locations are affected by the generator's coordinate-clamping convention. Correct, under-, and over-selection are defined relative to this nominal generating radius. For each generating radius, the loading matrix was constructed using the rule in Section~\ref{sec:supp-loading-matrix} with the corresponding ordered offset set. Changing \(k_{\mathrm{gen}}\) therefore also changes
the offset-indexed profiles.

The nonstationary local-neighbourhood generator was retained throughout. It used \(d=15\), with the Gaussian-process correlation range varying
from 1 to 4 grid cells across latitude. The common candidate set was \(k\in\{0,1,2,3,4\}\), so that \(s_{\max}=4\). This allowed both under- and over-selection and kept the \(k_{\mathrm{gen}}=3\) condition below the search boundary. SFGDPC was fitted using BIC, \texttt{normalize = 1}, a convergence tolerance of \(10^{-4}\), and a maximum of 500 alternating least-squares iterations. After generation, the observed coefficient matrix was column-centred
across spatial locations before fitting SFGDPC. NMSE was evaluated relative to this same centred observed field.

Each scenario used 100 Monte Carlo replications. Replication \(r\) was assigned seed \(271+r\), giving seeds 272 through 371, and the same replication seeds were used across scenarios. Table S3 summarizes the empirical BIC selection frequencies across the six sensitivity scenarios.

\begin{table}[ht]
\centering
\small
\setlength{\tabcolsep}{3.5pt}
\caption{Empirical BIC selection frequencies for the first SFGDPC component
over 100 Monte Carlo replications per scenario. Exact, under, and over denote
\(\widehat{k}=k_{\mathrm{gen}}\),
\(\widehat{k}<k_{\mathrm{gen}}\), and
\(\widehat{k}>k_{\mathrm{gen}}\), respectively.}
\label{tab:supp-bic-sensitivity}
\begin{tabular}{lccccccc}
\toprule
Scenario
& Grid
& \shortstack{Noise\\ratio}
& \shortstack{Nominal generating\\radius \(k_{\mathrm{gen}}\)}
& \shortstack{Selected radius\\(frequency)}
& Exact
& Under
& Over \\
\midrule
Baseline
& \(40\times40\)
& 0.15
& 2
& \(\widehat{k}=2\) (1.00)
& 1.00
& 0.00
& 0.00 \\

Small field
& \(20\times20\)
& 0.15
& 2
& \(\widehat{k}=1\) (1.00)
& 0.00
& 1.00
& 0.00 \\

Large field
& \(60\times60\)
& 0.15
& 2
& \(\widehat{k}=2\) (1.00)
& 1.00
& 0.00
& 0.00 \\

Higher noise
& \(40\times40\)
& 0.30
& 2
& \(\widehat{k}=2\) (1.00)
& 1.00
& 0.00
& 0.00 \\

Lag-zero loading
& \(40\times40\)
& 0.15
& 0
& \(\widehat{k}=0\) (1.00)
& 1.00
& 0.00
& 0.00 \\

\shortstack[l]{Larger generating\\radius}
& \(40\times40\)
& 0.15
& 3
& \(\widehat{k}=2\) (1.00)
& 0.00
& 1.00
& 0.00 \\
\bottomrule
\end{tabular}
\end{table}

All 600 replications completed without error, and all 3000 candidate fits reported convergence. No selected radius reached the search limit \(s_{\max}=4\). Manual and stored NMSE calculations agreed to within
\(2.8\times10^{-17}\).

The nominal radius was selected in every replication under the baseline, large-field, higher-noise, and lag-zero-loading conditions. In contrast, BIC underselected in every replication on the \(20\times20\) field and when \(k_{\mathrm{gen}}=3\). These results demonstrate exact recovery under several of the examined conditions, but they do not establish uniform or
asymptotically consistent recovery.

The field-size comparison should not be interpreted as a pure sample-size experiment. Under \(k_{\mathrm{gen}}=2\), the proportions of locations whose radius-two generating window extends beyond the observed grid are \(36.0\%\), \(19.0\%\), and \(12.9\%\) for the \(20\times20\),
\(40\times40\), and \(60\times60\) grids, respectively. Since the generator uses nearest-boundary clamping whereas the fitted SFGDPC method uses compacting self-padding, the smaller-field result combines reduced field size with a larger relative contribution from the
boundary region.

For \(k_{\mathrm{gen}}=3\), the outermost Chebyshev ring accounts for approximately \(5.3\%\) of the squared Frobenius norm of the deterministic loading matrix. Selection of \(\widehat{k}=2\) is therefore consistent with BIC identifying a more compact dominant reconstruction scale, although it remains under-selection under the strict generating-support definition.

\section{Sensitivity to the SFPCA Filter Radius}
\label{sec:supp-sfpca-radius}
The SFPCA comparison was examined over several fixed filter radii. For this sensitivity study, the supplied radii were restricted to values
satisfying \(1+2L\leq n-2L\), or equivalently
\(L\leq(n-1)/4\). This prespecified boundary-safe rule ensures that a nonempty region remains sufficiently separated from the grid boundary under the interior-diagnostic convention used by the SFPCA implementation. It is a design restriction for the present sensitivity analysis rather than a general admissibility condition for whole-grid SFPCA reconstruction. Under this rule, the largest examined integer radii were \(L=9\) on the \(40\times40\) grid and \(L=14\) on the \(60\times60\) grid.

The sensitivity analysis used 100 Monte Carlo replications under the nonstationary local-neighbourhood transfer model. Within each
replication, all candidate SFPCA radii were evaluated using the same simulated coefficient field. SFGDPC was fitted once per field because its estimates do not depend on the SFPCA filter radius. The SFPCA spectral smoothing parameter \(q\) was determined using the same data-adaptive rule implemented in \texttt{spca.automatic} as in
Section~\ref{sec:supp-sfarma}.

\begin{table}[ht]
\centering
\small
\setlength{\tabcolsep}{5pt}
\caption{Cumulative reconstruction NMSE in the SFPCA filter-radius
sensitivity analysis under the centred nonstationary local-neighbourhood
transfer design.}
\label{tab:supp-lmax-sweep}
\begin{tabular}{llcccc}
\toprule
Grid
& Method
& Radius \(L\)
& \(p=1\)
& \(p=2\)
& \(p=3\) \\
\midrule
\(40\times40\)
& SFPCA
& 3
& 0.1423 (0.0194)
& 0.1278 (0.0223)
& 0.1264 (0.0237) \\

\(40\times40\)
& SFPCA
& 5
& 0.1080 (0.0117)
& 0.0825 (0.0117)
& 0.0765 (0.0113) \\

\(40\times40\)
& SFPCA
& 7
& 0.0906 (0.0097)
& 0.0651 (0.0084)
& 0.0588 (0.0081) \\

\(40\times40\)
& SFPCA
& 9
& 0.0819 (0.0087)
& 0.0561 (0.0071)
& 0.0497 (0.0068) \\

\(40\times40\)
& SFGDPC
& --
& \textbf{0.0506 (0.0051)}
& \textbf{0.0334 (0.0035)}
& \textbf{0.0247 (0.0027)} \\

\midrule
\(60\times60\)
& SFPCA
& 9
& 0.0744 (0.0060)
& 0.0545 (0.0053)
& 0.0495 (0.0052) \\

\(60\times60\)
& SFPCA
& 11
& 0.0686 (0.0055)
& 0.0485 (0.0046)
& 0.0434 (0.0046) \\

\(60\times60\)
& SFPCA
& 13
& 0.0648 (0.0053)
& 0.0446 (0.0042)
& 0.0395 (0.0041) \\

\(60\times60\)
& SFPCA
& 14
& 0.0634 (0.0053)
& 0.0432 (0.0042)
& 0.0380 (0.0041) \\

\(60\times60\)
& SFGDPC
& --
& \textbf{0.0435 (0.0028)}
& \textbf{0.0260 (0.0018)}
& \textbf{0.0192 (0.0014)} \\
\bottomrule
\end{tabular}

\vspace{0.4em}
\begin{minipage}{0.94\textwidth}
\footnotesize
Notes: Entries are Monte Carlo means with standard deviations in parentheses
over 100 paired replications. Lower NMSE indicates better reconstruction;
boldface identifies the lowest mean within each grid and component count.
SFPCA used \(\beta=1\) and a \(200\times200\) frequency grid, and its
effective filter radius equalled the listed value in every replication.
SFGDPC was fitted once per simulated field using BIC selection over
\(k=0,\ldots,3\).
\end{minipage}
\end{table}
\FloatBarrier

The Monte Carlo mean SFPCA NMSE decreased monotonically as \(L\) increased on both grids. At the largest radius examined under the prespecified boundary-safe rule on the \(60\times60\) grid, the reductions in mean NMSE obtained by SFGDPC relative to SFPCA were approximately
$31.4\%$, $39.9\%$, and $49.4\%$ after one, two, and three components, respectively. 

\section{Representation Under the Stationary Control}
\label{sec:supp-spectral}
The stationary control also permits a population-level comparison of the representational structures of the two spatial methods. Although the loading matrix $B$ has full row rank, the noise-free coefficient field
is generated from a single scalar stationary latent process through a shift-invariant multivariate spatial kernel.

Ignoring finite-domain boundary effects, its spectral density matrix at spatial frequency $\bm{\theta}$ can be written as
$$
\bm{F}_{Z}(\bm{\theta}) =
g_f(\bm{\theta})
\bm{b}(\bm{\theta})
\bm{b}(\bm{\theta})^{\ast},
$$
where $g_f$ is the spectral density of the scalar latent field and
$$
\bm{b}(\bm{\theta}) =
\sum_{\bm{\ell}}
B_{\bm{\ell}}
\exp\left(
-\mathrm{i}\bm{\ell}^{\top}\bm{\theta}
\right).
$$
Consequently,
\(\operatorname{rank}\{\bm F_Z(\bm\theta)\}\leq1\) at each spatial frequency in the noise-free population model.

Thus, with unrestricted filters, one spatial dynamic principal component is sufficient to represent the stationary signal at the population level. SFPCA is therefore not intrinsically unable to represent the data-generating structure. Finite-sample estimation, filter truncation, and the different boundary conventions may contribute to the observed reconstruction difference; the present experiments do not isolate their contributions. The improvement of SFPCA with increasing \(L\) in the nonstationary sensitivity study in Table~\ref{tab:supp-lmax-sweep} is compatible with a contribution from filter truncation, but does not establish its contribution under the stationary control.